\documentclass[10pt,twocolumn,twoside]{IEEEtran}
\usepackage[latin9]{inputenc}
\usepackage{color}
\usepackage{amsmath}
\usepackage{amssymb}
\usepackage{graphicx}

\makeatletter

\usepackage[noadjust]{cite}

\usepackage[T1]{fontenc}
\usepackage{float}

\makeatletter

\floatstyle{ruled}
\newfloat{algorithm}{tbp}{loa}
\providecommand{\algorithmname}{Algorithm}
\floatname{algorithm}{\protect\algorithmname}

\newtheorem{theorem}{Theorem}
\newtheorem{proposition}[theorem]{Proposition}
\newtheorem{lemma}[theorem]{Lemma}

\newcounter{algo}

\usepackage[footnotesize]{caption}

\makeatother

\begin{document}

\title{\vspace{-0.55cm}\hspace{0.9cm} Decomposition by Partial Linearization:\newline
Parallel Optimization of Multi-Agent Systems}

\author{Gesualdo Scutari, Francisco Facchinei, Peiran Song, Daniel P. Palomar,
and Jong-Shi Pang \vspace{-0.6cm}%
\thanks{G. Scutari and P. Song are with the Dpt. of Electrical Eng., State
Univ. of New York at Buffalo, Buffalo, USA. F. Facchinei is with the
Dpt. of Computer, Control, and Management Eng., Univ. of Rome ``La
Sapienza'', Rome, Italy. J.-S. Pang is with the Dpt. of Industrial
and Systems Eng., Univ. of Southern California Viterbi School of Eng.,
Los Angeles, USA. D. Palomar is with the Dpt. of Electronic and Computer
Eng., Hong Kong Univ. of Science and Technology, Hong Kong. Emails:
\texttt{\textcolor{black}{<gesualdo, peiranso>@buffalo.edu}}\textcolor{black}{;}\texttt{\textcolor{black}{{}
facchinei@dis.uniroma1.it}}\textcolor{black}{;}\texttt{\textcolor{black}{{}
jongship@usc.edu}}\textcolor{black}{;}\texttt{\textcolor{black}{{} }}\textcolor{black}{and}\texttt{\textcolor{black}{{}
palomar@ust.hk}}.\texttt{\textcolor{red}{{} }}

Part of this work has been presented at the 5th Int. Conf. on Network
Games, Control and Optimization (NetGCooP 2011), Oct. 12-14, 2011,
\cite{ScutariPalomarFacchineiPang_NETGCOP11}. %
}}
\maketitle
\begin{abstract}
We propose a novel decomposition framework for the distributed optimization
of \emph{general} nonconvex sum-utility functions arising naturally
in the system design of wireless multi-user interfering systems. Our
main contributions are: i) the development of the first class of (inexact)
\emph{Jacobi best-response} algorithms with provable convergence,
where all the users simultaneously and iteratively solve a suitably
convexified version of the original sum-utility optimization problem;
ii) the derivation of a general dynamic pricing mechanism that provides
a unified view of existing pricing schemes that are based, instead,
on heuristics; and iii) a framework that can be easily particularized
to well-known applications, giving rise to very efficient practical
(Jacobi or Gauss-Seidel) algorithms that outperform  existing ad-hoc
methods proposed for very specific problems. 
Interestingly, our framework contains as special cases well-known
gradient algorithms for nonconvex sum-utility problems, and many block-coordinate
descent schemes for convex functions.\vspace{-0.4cm}
\end{abstract}
\IEEEpeerreviewmaketitle

\section{Introduction \vspace{-0.1cm}}

\IEEEPARstart{W}{ireless} networks are composed of users that
may have different objectives and generate interference, when no multiplexing
scheme is imposed a priori to regulate the transmissions; examples
are peer-to-peer, ad-hoc, and cognitive radio systems. A usual and
convenient way of designing such multiuser systems is by optimizing
the ``social function'', i.e., the (weighted) sum of the users'
objective functions. Since centralized solution methods are too demanding
in most applications, the main difficulty of this formulation lies
in performing the optimization in a distributed manner with limited
signaling among the users. When the social problem is a sum-separable
\emph{convex} programming, many distributed methods have been proposed,
based on primal and dual decomposition techniques; see, e.g., \cite{Palomar-Chiang_ACTran07-Num,ChiangLowCalderbankDoyle_ProcIEEE07,Bertsekas_Book-Parallel-Comp}
and references therein. In this paper we address the more frequent
and difficult case in which the social function is nonconvex. \textcolor{black}{It
is well known that the problem of finding a global minimum of the
social function is, in general, NP hard (see e.g. \cite{Luo-Zhang}),
and centralized solution methods (e.g., based on combinatorial approaches)
are too demanding in most applications. As a consequence, recent research
efforts have been focused on finding efficiently high quality suboptimal
solutions via easy-to-implement (possibly) distributed algorithms.
A recent survey on nonconvex resource allocation problems in interfering
networks modeled as Gaussian Interference Channels (ICs) is \cite{HongLuoTutorial12}.}

In an effort to obtain distributed albeit suboptimal algorithms a
whole spectrum of approaches has been explored, trying to balance
practical effectiveness and coordination requirements. At one end
of the spectrum we find game-theoretical approaches, where users in
the network are modeled as players that greedily optimize their own
objective function. Game-theoretical models for power control problems
over ICs have been proposed in \cite{Yu-Ginis-Cioffi_jsac02,Luo-Pang_IWFA-Eurasip,Scutari-Palomar-Barbarossa_SP08_PI,Scutari-Palomar- Barbarossa_AIWFA_IT08,Huang-Cendrillon-Chiang-Moonen_SP07}
and \cite{Scutari-Palomar-Barbarossa_JSAC08,Scutari-Palomar-   Barbarossa_SP08_IWFA,Scutari- Palomar_SP09_CR_GTMIMO}
for SISO and MISO/MIMO systems, respectively. Two recent tutorials
on the subject are \cite{Larsson-Jorswieck-Lindblom-Mochaourab_SPMag_sub09,Leshem- Zehavip_SPMag_sub09},
while recent contributions using the more general mathematical theory
of Variational Inequalities \cite{Facchinei-Pang_FVI03} are \cite{Pang-Scutari-Palomar-Facchinei_SP_10,Scutari-Palomar- Facchinei-Pang_SPMag09,Scutari-Palomar-Facchinei-Pang_SPMag10}.
The advantage of game-theoretic methods is that they lead to distributed
implementations (only local channel information is required at each
user); however they converge to Nash equilibria that in general are
not even stationary solutions of the nonconvex social problem. In
contrast, other methods aim at reaching stationary solutions of the
nonconvex social problem, at the cost of more signaling and coordination.
\emph{Sequential} decomposition algorithms were proposed in \cite{Huang-Berry-Honig_JSAC06,Wang-Krunz-Cui_JSTSP08,SchmidtShiBerryHonigUtschick-SPMag,KimGiannakisIT11}
for the sum-rate maximization problem over SISO/MIMO ICs, and in \cite{GrippoSciandroneOPL}
for more general (nonconvex) functions. In these algorithms, only
one agent at a time is allowed to update his optimization variables;
a fact that in large scale networks may lead to excessive communication
overhead and slow convergence. 

The aim of this paper is instead the study of more appealing \emph{simultaneous}
\emph{distributed} methods for \emph{general} nonconvex sum-utility
problems, where \emph{all} users can update their variables at the
same time. The design of such algorithms with provable convergence
is much more difficult, as also witnessed by the scarcity of results
available in the literature. Besides the application of the classical
gradient projection algorithm to the sum-rate maximization problem
over MIMO ICs \cite{Ye-Blum_SP03}, parallel iterative algorithms
(with message passing) for DSL/ad-hoc SISO networks and MIMO broadcast
interfering channels were proposed in \cite{Chiang-WeiTan-PalomarOneil-Julian_ITWC-GP,TsiaflakisMoonenTSP08,Papandriopoulos_EvansIT09}
and \cite{ShiRazaviyaynLuoHe-TSP11}, respectively. Unfortunately,
the gradient schemes \cite{Ye-Blum_SP03} suffer from slow convergence
and do not exploit any degree of convexity that might be present in
the sum-utility function; \cite{Chiang-WeiTan-PalomarOneil-Julian_ITWC-GP,TsiaflakisMoonenTSP08,Papandriopoulos_EvansIT09}
hinge crucially on the special log-structure of the users' rate functions;
and \cite{ShiRazaviyaynLuoHe-TSP11} is based on the connection with
a weighted MMSE problem. This makes \cite{Chiang-WeiTan-PalomarOneil-Julian_ITWC-GP,TsiaflakisMoonenTSP08,Papandriopoulos_EvansIT09,ShiRazaviyaynLuoHe-TSP11}
 not applicable to different classes of sum-utility problems.

Building on the idea first introduced in \cite{ScutariPalomarFacchineiPang_NETGCOP11},
the main contribution of this paper is to propose a new decomposition
method that: \emph{i)} converges to stationary points of a large class
of (nonconvex) social problems, encompassing most sum-utility functions
of practical interest (including functions of complex variables);
\emph{ii)} decomposes well across the users, resulting in the \emph{parallel}
solution of \emph{convex} subproblems, one for each user; \emph{iii)
}converges also if the users' subproblems are solved in an inexact
way; and\emph{ iv)} contains as special case the gradient algorithms
for nonconvex sum-utility problems, and many block-coordinate descent
schemes for convex functions. Moreover, the proposed framework can
be easily particularized to well-known applications, such as \cite{Huang-Berry-Honig_JSAC06,Wang-Krunz-Cui_JSTSP08,SchmidtShiBerryHonigUtschick-SPMag,ShiSchmidtBerryHonigUtschick-ICC09,Papandriopoulos_EvansIT09,KimGiannakisIT11},
giving rise in a unified fashion to distributed simultaneous algorithms
that outperform existing \emph{ad-hoc} methods both theoretically
and numerically. We remark that while we follow the seminal ideas
put forward in \cite{ScutariPalomarFacchineiPang_NETGCOP11}, in this
paper, besides providing full proofs of the results in \cite{ScutariPalomarFacchineiPang_NETGCOP11},
we i) consider a much wider class of social-problems and (possibly
inexact) algorithms, including \cite{ScutariPalomarFacchineiPang_NETGCOP11}
as special cases, ii) discuss in detail the case of functions of complex
variables, and iii) compare numerically to state-of-the-art alternative
methods. To the best of our knowledge, this paper is the first attempt
toward the development of decomposition techniques for \emph{general
}nonconvex sum-utility problems that allow \emph{distributed simultaneous
}(possibly\emph{ inexact})\emph{ best-response}-based updates among
the users.

On one hand, our approach draws on the Successive Convex Approximation
(SCA) paradigm, but relaxes the key requirement that the convex approximation
must be a tight global upper bound of the social function, as required
instead in \cite{MarksWright78,Chiang-WeiTan-PalomarOneil-Julian_ITWC-GP,RazaviyaynHongLuo_subOct12arxiv}
(see Sec \ref{sec:Extensions-and-Generalizations} for a detailed
comparison with \cite{MarksWright78,Chiang-WeiTan-PalomarOneil-Julian_ITWC-GP,RazaviyaynHongLuo_subOct12arxiv}).
This represents a turning point in the design of distributed SCA-based
methods, since up to date, finding such an upper bound convex approximation
for sum-utility functions having no specific structure (as, e.g.,
\cite{Papandriopoulos_EvansIT09,TsiaflakisMoonenTSP08,Chiang-WeiTan-PalomarOneil-Julian_ITWC-GP,KimGiannakisIT11,Ye- Blum_SP03,ShiRazaviyaynLuoHe-TSP11})
has been an elusive task. 

On the other hand, our method also sheds new light on widely used
pricing mechanisms: indeed, our scheme can be viewed as a dynamic
pricing algorithm where the pricing rule derives from a deep understanding
of the problem characteristics and is not obtained on an ad-hoc basis,
as instead in \cite{Huang-Berry-Honig_JSAC06,Wang-Krunz-Cui_JSTSP08,SchmidtShiBerryHonigUtschick-SPMag,ShiSchmidtBerryHonigUtschick-ICC09,KimGiannakisIT11}.
We conclude this review by mentioning the recent work \cite{HongLiLiuLuo_subOct12arxiv},
where the authors, developing ideas contained in \cite{ShiRazaviyaynLuoHe-TSP11,RazaviyaynHongLuo_subOct12arxiv},
proposed parallel schemes based of the SCA idea that are applicable
(only) to the class of sum-utility problems for which a connection
with a MMSE formulation can be established. Note that \cite{RazaviyaynHongLuo_subOct12arxiv}
\cite{HongLiLiuLuo_subOct12arxiv}, which share some ideas with our
approach, appeared after \cite{ScutariPalomarFacchineiPang_NETGCOP11}.

The rest of the paper is organized as follows. Sec. \ref{sec:Problem-Formulation}
introduces the sum-utility optimization problem along with some motivating
examples. Sec. \ref{sec:Distributed-Dynamic-Pricing} presents our
novel decomposition mechanism based on partial linearizations; the
algorithmic framework is described in Sec. \ref{sub:Jacobi-Distributed-Pricing}.
Sec. \ref{sec:The-Complex-Case} extends our results to sum-utility
problems in the complex domain; further generalizations are discussed
in Sec. \ref{sec:Extensions-and-Generalizations}. In Sec. \ref{sec:Applications}
we apply our new algorithms to some resource allocation problems over
SISO and MIMO ICs, and compare their performance with the state-of-the-art
decomposition schemes. Finally, Sec. \ref{sec:Conclusion} draws some
conclusions.\vspace{-0.2cm}

\section{Problem Formulation\label{sec:Problem-Formulation}}

We consider the design of a multiuser system composed of $I$ coupled
users $\mathcal{I}\triangleq\{1,\ldots,I\}$. Each user $i$ makes
decisions on his own $n_{i}$-dimensional real strategy vector $\mathbf{x}_{i}$,
which belongs to the feasible set $\mathcal{K}_{i}$; the vector variables
of the other users is denoted by $\mathbf{x}_{-i}\!\triangleq\!(\mathbf{x}_{j})_{j\neq i}\!\in\!\mathcal{K}_{-i}\!\triangleq\!\prod_{j\neq i}\mathcal{K}_{j}$;
the users' strategy profile is $\mathbf{x}\!=\!(\mathbf{x}_{i})_{i=1}^{I}$,
and the joint strategy set of the users is $\mathcal{K}\triangleq\prod_{j\in\mathcal{I}}\mathcal{K}_{j}$.
The system design is formulated as

\begin{equation}
\begin{array}{ll}
\underset{\mathbf{x}}{\mbox{minimize}} & {\displaystyle U(\mathbf{x})\triangleq{\sum_{\ell\in\mathcal{I}_{f}}}}\, f_{\ell}(\mathbf{x})\\[5pt]
\text{subject to} & \mathbf{x}_{i}\in{\cal K}_{i},\quad\forall i\in\mathcal{I},
\end{array}\label{eq:social problem}
\end{equation}
with $\mathcal{I}_{f}\triangleq\{1,\ldots,I_{f}\}$. Observe that,
in principle, the set $\mathcal{I}_{f}$ of objective functions is
different from the set $\mathcal{I}$ of users; we show shortly how
to explore this extra degree of freedom to good effect. Of course,
(\ref{eq:social problem}) contains the most common case where there
is exactly one function for each user, i.e. $I=I_{f}$.

\noindent\textbf{Assumptions.} We make the following blanket assumptions: 

\noindent A1) Each $\ensuremath{\mathcal{K}_{i}}$ is closed and convex;

\noindent A2) Each \emph{$\ensuremath{f_{i}}$} is continuously differentiable
on $\mathcal{K}$;

\noindent A3) Each $\nabla_{\mathbf{x}}f_{i}$ is Lipschitz continuous
on $\mathcal{K}$, with constant $L_{\nabla f_{i}}$; let $L_{\nabla U}\triangleq\sum_{i}L_{\nabla f_{i}}$;

\noindent A4) The lower level set $\mathcal{L}(\mathbf{x}^{0})\triangleq\{\mathbf{x}\in\mathcal{K}\,:\, U(\mathbf{x})\leq U(\mathbf{x}^{0})\}$
of the social function $U$ is compact for some $\mathbf{x}^{0}\in\mathcal{K}$.

The assumptions above are quite standard and are satisfied by a large
class of problems of practical interest. In particular, condition
A4 guarantees that the social problem has a solution, even when the
feasible $\mathcal{K}$ is not bounded; if $\mathcal{K}$ is bounded
A4 is trivially satisfied. A sufficient condition for A4 when $\mathcal{K}$
is not necessarily bounded is that $U$ be coercive {[}i.e., $U(\mathbf{x})\rightarrow+\infty$
as $\|\mathbf{x}\|\rightarrow+\infty$, with $\mathbf{x}\in\mathcal{K}${]}.
Note that, differently from classical Network Utility Maximization
(NUM) problems, here we do not assume any convexity of the functions
$f_{\ell}$, thus, (\ref{eq:social problem}) is a nonconvex minimization
problem. For the sake of simplicity, in (\ref{eq:social problem})
we assume that the users' strategies are real vectors; in Sec. \ref{sec:The-Complex-Case},
we extend our framework to complex matrix strategies, to cover also
the design of MIMO systems. 

\noindent\emph{A motivating example}. The social problem (\ref{eq:social problem})
is general enough to encompass many sum-utility problems of practical
interest. It also includes well-known utility functions studied in
the literature; an example is given next. Consider an $N$-parallel
Gaussian IC composed of $I$ active users, and let 
\[
r_{i}(\mathbf{p}_{i},\mathbf{p}_{-i})\triangleq{\displaystyle {\sum_{k=1}^{N}}}\log\left(1+\frac{\left|H_{ii}\left(k\right)\right|^{2}p_{ik}}{\sigma_{ik}^{2}+\sum_{j\neq i}\left|H_{ij}\left(k\right)\right|^{2}p_{jk}}\right)
\]
be the maximum achievable rate on link $i$, where $\mathbf{p}_{i}\triangleq(p_{ik})_{k=1}^{N}$
denotes the power allocation of user $i$ over the $N$ parallel channels,
$\mathbf{p}_{-i}\triangleq(\mathbf{p}_{j})_{j\neq i}$ is the power
profile of all the other users $j\neq i$, $\left|H_{ij}\left(k\right)\right|^{2}$
is the gain of the channel between the $j$-th transmitter and the
$i$-th receiver, $\sigma_{ik}^{2}$ is the variance of the thermal
noise over carrier $k$ at the receiver $i,$ and $\sum_{j\neq i}\left|H_{ij}\left(k\right)\right|^{2}p_{jk}$
represents the multiuser interference generated by the users $j\neq i$
at the receiver $i$. Each transmitter $i$ is subject to the power
constraints $\mathbf{p}_{i}\in\mathcal{P}_{i}$, with
\begin{equation}
\mathcal{P}_{i}\triangleq\left\{ \mathbf{p}_{i}\in\mathbb{R}_{+}^{N}\,:\,\mathbf{W}_{i}\mathbf{p}_{i}\leq\mathbf{I}_{i}^{\max}\right\} ,\label{eq:set_power_constraints}
\end{equation}
where the inequality, with given $\mathbf{I}_{i}^{\max}\in\mathbb{R}_{+}^{m_{i}}$
and $\mathbf{W}_{i}\in\mathbb{R}_{+}^{m_{i}\times N}$ is intended
to be  component-wise. Note that the linear (vector) constraints in
(\ref{eq:set_power_constraints}) are general enough to model classical
power budget constraints and different interference constraints, such
as spectral mask or interference-temperature limits. Finally, let
$\theta_{i}:\mathbb{R}_{+}\rightarrow\mathbb{R}$ be the utility functions
of the users' rates. The system design can then be formulated as

\begin{equation}
\hspace{-0.2cm}\begin{array}{ll}
\underset{\mathbf{p}_{1},\ldots,\mathbf{p}_{I}}{\mbox{maximize}} & {\displaystyle {\sum_{i\in\mathcal{I}}}}\,\theta_{i}\left(r_{i}(\mathbf{p}_{i},\mathbf{p}_{-i})\right)\\[5pt]
\text{subject to} & \mathbf{p}_{i}\in\mathcal{P}_{i},\quad\forall i\in\mathcal{I}.
\end{array}\label{eq:SISO_formulation}
\end{equation}
 Note that (\ref{eq:SISO_formulation}) is an instance of (\ref{eq:social problem}),
with $I_{f}=I$; moreover assumptions A1-A4 are satisfied if the utility
functions $\theta_{i}(x)$ are i) concave and nondecreasing on $\mathbb{R}_{+}$,
and ii) continuously differentiable with Lipschitz gradients. Interestingly,
this class of functions $\theta_{i}(x)$ includes many well-known
special cases studied in the literature, such as the weighted sum-rate
function, the harmonic mean of the rates, the geometric mean of (one
plus) the rates, etc.; see, e.g., \cite{Huang-Berry-Honig_JSAC06,Wang-Krunz-Cui_JSTSP08,CendrillonYuMoonenVerlindenBostoen_TCOM06,HongLuoTutorial12}. 

Since the class of problems (\ref{eq:social problem}) is in general
nonconvex (generally NP hard \cite{Luo-Zhang}), the focus of this
paper is on the design of \emph{distributed }solution methods for
computing stationary solutions (possibly local minima) of (\ref{eq:social problem}).
Our major goal is to devise \emph{simultaneous best-response }schemes
fully decomposed across the users, meaning that all the users can
solve \emph{in parallel} a sequence of convex problems while converging
to a stationary solution of the original nonconvex problem.

\section{A New Decomposition Technique\label{sec:Distributed-Dynamic-Pricing}}

We begin by introducing an informal description of our new algorithms
that sheds light on the core idea of the novel decomposition technique
and establishes the connection with classical descent gradient-based
schemes. This will also explain why our scheme is expected to outperform
current gradient methods. A formal description of the proposed algorithms
along with their main properties is given in Sec. \ref{sub:Jacobi-Distributed-Pricing}
for the real case, and in Sec. \ref{sec:The-Complex-Case} for the
complex case.\vspace{-0.3cm}

\subsection{What do conditional gradient methods miss?\label{sub:GP_intepretation}}

A classical approach to solve a nonconvex problem like (\ref{eq:social problem})
would be using some well-known gradient-based descent scheme. A simple
way to generate a (feasible) descent direction is for example using
the conditional gradient method (also called Frank-Wolfe method) \cite{Bertsekas_Book-Parallel-Comp}:
given the current iterate $\mathbf{x}^{n}=(\mathbf{x}_{i}^{n})_{i=1}^{I}$,
the next feasible vector $\mathbf{x}^{n+1}$ is given by 
\begin{equation}
\begin{array}{l}
\mathbf{x}^{n+1}=\mathbf{x}^{n}+\gamma^{n}\,\mathbf{d}^{n}\end{array}\label{eq:feasible_dir}
\end{equation}
where $\mathbf{d}^{n}\triangleq\overline{\mathbf{x}}^{n}-\mathbf{x}^{n}$,
$\overline{\mathbf{x}}^{n}=(\overline{\mathbf{x}}_{i}^{n})_{i=1}^{I}$
is the solution of the following set of convex problems (one for each
user):
\begin{equation}
\overline{\mathbf{x}}_{i}^{n}=\underset{\mathbf{x}_{i}\in\mathcal{K}_{i}}{\text{argmin}}\left\{ \nabla_{\mathbf{x}_{i}}U\left(\mathbf{x}^{n}\right)^{T}\left(\mathbf{x}_{i}-\mathbf{x}_{i}^{n}\right)\right\} ,\label{eq:lin_prob}
\end{equation}
for all $i\in\mathcal{I}$, and $\gamma^{n}\in(0,1]$ is the step-size
of the algorithm that needs to be properly chosen to guarantee convergence. 

Looking at (\ref{eq:lin_prob}) one infers that gradient methods are
based on solving a sequence of parallel convex problems, one for each
user, obtained by linearizing the \emph{whole }utility function $U(\mathbf{x})$
around $\mathbf{x}^{n}$, a fact that does not exploit any ``nice''
structure that the original problem may potentially have. 

At the basis of the proposed decomposition techniques, there is instead
the attempt to properly exploit any degree of convexity that might
be present in the social function. To capture this idea, for each
user $i\in\mathcal{I}$, let $\mathcal{S}_{i}\subseteq\mathcal{I}_{f}$
be the set of indices of all the functions $f_{j}(\mathbf{x}_{i},\mathbf{x}_{-i})$
that are convex in $\mathbf{x}_{i}$ on $\mathcal{K}_{i}$, for any
given $\mathbf{x}_{-i}\in\mathcal{K}_{-i}$: \vspace{-0.2cm}

\begin{equation}
\mathcal{S}_{i}\triangleq\left\{ j\in\mathcal{I}_{f}\!:\! f_{j}(\bullet,\mathbf{x}_{-i})\,\mbox{is convex on }\mathcal{K}_{i},\forall\mathbf{x}_{-i}\in\mathcal{K}_{-i}\right\} \label{eq:C)i_set-1}
\end{equation}
and let $\mathcal{C}_{i}\subseteq\mathcal{S}_{i}$ be a given subset
of $\mathcal{S}_{i}$. The idea is to preserve the convex structure
of the functions in $\mathcal{C}_{i}$ while linearizing the rest.
Note that we allow the possibility that $\mathcal{S}_{i}=\emptyset$,
even if we ``hope'' that $\mathcal{S}_{i}\neq\emptyset$, and actually
this latter case occurs in most of the applications of interest, see
Sec. \ref{sec:Applications}. For each user $i\in\mathcal{I},$ we
can introduce the following convex approximation of $U(\mathbf{x})$
around $\mathbf{x}^{n}\in\mathcal{K}$:
\begin{eqnarray}
\widetilde{f}_{\mathcal{C}_{i}}(\mathbf{x}_{i};\mathbf{x}^{n}) & \triangleq & {\displaystyle {\sum_{j\in\mathcal{C}_{i}}}}f_{j}(\mathbf{x}_{i},\,\mathbf{x}_{-i}^{n})+\boldsymbol{\pi}_{\mathcal{C}_{i}}(\mathbf{x}^{n})^{T}(\mathbf{x}_{i}-\mathbf{x}_{i}^{n})\nonumber \\
 &  & +\dfrac{{\tau}_{i}}{2}\,\left(\mathbf{x}_{i}-\mathbf{x}_{i}^{n}\right)^{T}\mathbf{H}_{i}(\mathbf{x}^{n})\left(\mathbf{x}_{i}-\mathbf{x}_{i}^{n}\right)\label{eq:convex_approx_of_fi_on_Ci}
\end{eqnarray}
with\vspace{-0.1cm} 
\begin{equation}
\boldsymbol{\pi}_{\mathcal{C}_{i}}(\mathbf{x}^{n})\triangleq\sum_{j\in\mathcal{C}_{-i}}\!\left.\nabla_{\mathbf{x}_{i}}f_{j}(\mathbf{x})\right|_{\mathbf{x}=\mathbf{x}^{n}},\vspace{-0.1cm}\label{eq:pricing_pi_C}
\end{equation}
where $\mathcal{C}_{-i}\triangleq\mathcal{I}_{f}\backslash\mathcal{C}_{i}$
is the complement of $\mathcal{C}_{i}$, $\tau_{i}$ is a given \textcolor{black}{nonnegative
}constant, and $\mathbf{H}_{i}(\mathbf{x}^{n})$ is an $n_{i}\times n_{i}$
uniformly positive definite matrix (possibly dependent on $\mathbf{x}^{n})$,
i.e. $\mathbf{H}_{i}(\mathbf{x}^{n})-c_{H_{i}}\mathbf{I}\succeq\mathbf{0}$,
for some positive $c_{H_{i}}$. For notational simplicity, we omitted
in $\widetilde{f}_{\mathcal{C}_{i}}(\mathbf{x}_{i};\mathbf{x}^{n})$
the dependence on $\tau_{i}$ and $\mathbf{H}_{i}(\mathbf{x}^{n})$.
Note that in (\ref{eq:convex_approx_of_fi_on_Ci}), we added a proximal-like
regularization term, in order to relax the convergence conditions
of the resulting algorithm or enhance the convergence speed (cf. Sec.
\ref{sub:Jacobi-Distributed-Pricing}). \textcolor{black}{A key feature
of $\widetilde{f}_{\mathcal{C}_{i}}$ we will always require is that
$\widetilde{f}_{\mathcal{C}_{i}}(\bullet;\mathbf{x})$ be uniformly
strongly convex. By this we mean the following. Let $c_{\tau_{i}}(\mathbf{x})$
be the constant of strong convexity of $\widetilde{f}_{\mathcal{C}_{i}}(\bullet;\mathbf{x})$.
We require that $ $
\begin{equation}
c_{\tau_{i}}\triangleq\inf_{\mathbf{x}\in\mathcal{K}}c_{\tau_{i}}(\mathbf{x})>0.\label{eq:c_tau_i}
\end{equation}
Note that this }\textcolor{black}{\emph{is}}\textcolor{black}{{} }\textcolor{black}{\emph{not}}\textcolor{black}{{}
}\textcolor{black}{\emph{an additional assumption}}\textcolor{black}{,
but just a requirement on the way $\tau_{i}$ is chosen.}\textcolor{black}{\emph{
}}\textcolor{black}{Under the uniformly positive definiteness of $\mathbf{H}_{i}(\mathbf{x}^{n})$,
condition (\ref{eq:c_tau_i}) is always satisfied if $\tau_{i}>0$;
however it is also satisfied with $\tau_{i}=0$ if $\sum_{j\in\mathcal{C}_{i}}f_{j}(\bullet,\mathbf{x}_{-i})$
is uniformly strongly convex on $\mathcal{K}_{-i}$; a fact that occurs
in many applications, see, e.g., Sec. \ref{sec:Applications}.}

Associated with each $\widetilde{f}_{\mathcal{C}_{i}}(\mathbf{x}_{i};\mathbf{x}^{n})$
we can define the following ``best response'' map that resembles
(\ref{eq:lin_prob}): 
\begin{equation}
\widehat{\mathbf{x}}_{\mathcal{C}_{i}}(\mathbf{x}^{n},\tau_{i})\triangleq\underset{\mathbf{x}_{i}\in\mathcal{K}_{i}}{\mbox{argmin}\,}\widetilde{f}_{\mathcal{C}_{i}}(\mathbf{x}_{i};\mathbf{x}^{n}).\label{eq:decoupled_problem_i}
\end{equation}
 Note that, in the setting above, $\widehat{\mathbf{x}}_{\mathcal{C}_{i}}(\mathbf{x}^{n},\tau_{i})$
is always well-defined, since the optimization problem in (\ref{eq:decoupled_problem_i})
is strongly convex and thus has a unique solution. Given (\ref{eq:decoupled_problem_i}),
we can introduce the best-response mapping of the users, defined as
\begin{equation}
\mathcal{K}\ni\mathbf{y}\mapsto\widehat{\mathbf{x}}_{\mathcal{C}}(\mathbf{y},\boldsymbol{{\tau}})\triangleq\left(\widehat{\mathbf{x}}_{\mathcal{C}_{i}}(\mathbf{y},\tau_{i})\right)_{i=1}^{I};\label{eq:best-reponse}
\end{equation}
and also set $\boldsymbol{\tau}\triangleq(\tau_{i})_{i=1}^{I}$. The
proposed search direction $\mathbf{d}^{n}$ at point $\mathbf{x}^{n}$
in (\ref{eq:feasible_dir}) becomes then $\widehat{\mathbf{x}}_{\mathcal{C}}(\mathbf{x}^{n},\boldsymbol{{\tau}})-\mathbf{x}^{n}$.
The challenging question now is whether such direction is still a
descent direction for the function $U$ at $\mathbf{x}^{n}$ and how
to choose the free parameters (such as $\tau_{i}$'s, $\gamma^{n}$'s,
and $\mathbf{H}_{i}(\mathbf{x}^{n})$'s) in order to guarantee convergence
to a stationary solution of the original nonconvex sum-utility problem.
These issues are addressed in the next sections. \vspace{-0.1cm}

\subsection{Properties of the best-response mapping $\widehat{\mathbf{x}}_{\mathcal{C}}(\mathbf{y},\boldsymbol{{\tau}})$\label{sub:Properties-of-the_BRmapping}}

Before introducing a formal description of the proposed algorithms,
we derive next some key properties of the best-response map $\widehat{\mathbf{x}}_{\mathcal{C}}(\mathbf{y},\boldsymbol{{\tau}})$,
which shed light on how to choose the free parameters in (\ref{eq:decoupled_problem_i})
and prove convergence. 

\begin{proposition}\label{Prop_x_y}Given the social problem (\ref{eq:social problem})
under A1)-A4), \textcolor{black}{suppose that}\textcolor{blue}{{} }\textcolor{black}{each
$\mathbf{H}_{i}(\mathbf{x})-c_{H_{i}}\mathbf{I}\succeq\mathbf{0}$
for all $\mathbf{x}\in\mathcal{K}$}\textcolor{blue}{{} }\textcolor{black}{and
some $c_{H_{i}}>0$, and}\textcolor{blue}{{} }\textcolor{black}{$(c_{\tau_{i}})_{i=1}^{I}>\mathbf{0}$.}
Then the mapping\emph{ $\mathcal{K}\ni\mathbf{y}\mapsto\widehat{\mathbf{x}}(\mathbf{y},\boldsymbol{{\tau}})$
}has the following properties:\emph{ }

\noindent (a)\emph{ $\widehat{\mathbf{x}}_{\mathcal{C}}(\mathbf{\bullet},\boldsymbol{{\tau}})$}
is Lipschitz\emph{ }continuous on\emph{ $\mathcal{K}$, }i.e., there
exists a positive constant $\hat{{L}}$ such that\emph{ 
\begin{equation}
\left\Vert \widehat{\mathbf{x}}_{\mathcal{C}}(\mathbf{y},\boldsymbol{{\tau}})-\widehat{\mathbf{x}}_{\mathcal{C}}(\mathbf{z},\boldsymbol{{\tau}})\right\Vert \leq\,\hat{{L}}\,\left\Vert \mathbf{y}-\mathbf{z}\right\Vert ,\quad\forall\mathbf{y},\mathbf{z}\in\mathcal{K};\vspace{-0.3cm}\label{eq:Lipt_x_map}
\end{equation}
}

\noindent(b)\emph{ }The set of the fixed-points of\emph{ }$\widehat{\mathbf{x}}_{\mathcal{C}}(\mathbf{\bullet},\boldsymbol{{\tau}})$\emph{
}coincides with the set of stationary solutions of the social problem\emph{
}(\ref{eq:social problem})\emph{;} therefore\emph{ }$\widehat{\mathbf{x}}_{\mathcal{C}}(\mathbf{y},\boldsymbol{{\tau}})$\emph{
}has a fixed-point; 

\noindent(c)\emph{ }For every given\emph{ $\mathbf{y}\in\mathcal{K}$,
}the vector\emph{ }$\widehat{\mathbf{x}}_{\mathcal{C}}(\mathbf{y},\boldsymbol{{\tau}})-\mathbf{y}$\emph{
}is a descent direction of the social function\emph{ $U(\mathbf{x})$
}at\emph{ $\mathbf{y}$ }such that\emph{ 
\begin{equation}
\left(\widehat{\mathbf{x}}_{\mathcal{C}}(\mathbf{y},\boldsymbol{{\tau}})-\mathbf{y}\right)^{T}\,\nabla_{\mathbf{x}}U(\mathbf{y})\leq-c\,\left\Vert \widehat{\mathbf{x}}_{\mathcal{C}}(\mathbf{y},\boldsymbol{{\tau}})-\mathbf{y}\right\Vert ^{2},\label{eq:descent_direction}
\end{equation}
}for some positive constant $c\geq c_{\boldsymbol{{\tau}}}$, with\vspace{-0.1cm}\textcolor{black}{
\begin{equation}
c_{\boldsymbol{{\tau}}}\triangleq\min_{i\in\mathcal{I}}\left\{ c_{\tau_{i}}\right\} .\vspace{-0.3cm}\label{eq:c_tau}
\end{equation}
}

\noindent(d) If $\nabla_{\mathbf{x}}U(\mathbf{x})$ is bounded on
$\mathcal{K}$, then there exists a finite constant $\alpha>0$ such
that\vspace{-0.1cm} 
\begin{equation}
\left\Vert \widehat{\mathbf{x}}_{\mathcal{C}}(\mathbf{y},\boldsymbol{{\tau}})-\mathbf{y}\right\Vert \leq\,\alpha,\quad\forall\mathbf{y}\in\mathcal{K}.\label{eq:bounded_distance}
\end{equation}

\end{proposition}

\begin{proof}See Appendix A. \end{proof}\smallskip

Proposition \ref{Prop_x_y} makes formal the idea introduced in Sec.
\ref{sub:GP_intepretation} and thus paves the way to the design of
distributed best-response-like algorithms for (\ref{eq:social problem})
based on $\widehat{\mathbf{x}}_{\mathcal{C}}(\mathbf{\bullet},\boldsymbol{{\tau}})$.
Indeed, the inequality (\ref{eq:descent_direction}) states that either
$\left(\widehat{\mathbf{x}}_{\mathcal{C}}(\mathbf{x}^{n})-\mathbf{x}^{n}\right)^{T}\,\nabla_{\mathbf{x}}U(\mathbf{x}^{n})<0$
or $\widehat{\mathbf{x}}_{\mathcal{C}}(\mathbf{x}^{n})=\mathbf{x}^{n}$.
In the former case, $\mathbf{d}^{n}\triangleq\widehat{\mathbf{x}}_{\mathcal{C}}(\mathbf{x}^{n})-\mathbf{x}^{n}$
is a descent direction of $U(\mathbf{x})$ at $\mathbf{x}^{n}$; in
the latter case, $\mathbf{x}^{n}$ is a fixed-point of the mapping
$\widehat{\mathbf{x}}_{\mathcal{C}}(\mathbf{\bullet},\boldsymbol{{\tau}})$
and thus a stationary solution of the original nonconvex problem (\ref{eq:social problem})
{[}Prop. \ref{Prop_x_y} (b){]}. 

Quite interestingly, we can also provide a characterization of the
fixed-points of $\widehat{\mathbf{x}}_{\mathcal{C}}(\mathbf{y},\boldsymbol{{\tau}})$
{[}and thus the stationary solutions of (\ref{eq:social problem}){]}
in terms of Nash equilibria of a game with a proper pricing mechanism.
Formally, we have the following. \smallskip

\begin{proposition}\label{Lemma_NE} Any fixed-point $\mathbf{x}^{\star}$
of $\widehat{\mathbf{x}}_{\mathcal{C}}(\mathbf{\bullet},\boldsymbol{{\tau}})$
is a Nash equilibrium of the game where each user $i\in\mathcal{I}$
solves the following priced convex optimization problem: given $\mathbf{x}_{-i}$,
\begin{equation}
\underset{\mathbf{x}_{i}\in\mathcal{K}_{i}}{\mbox{min}\,}{\displaystyle {\sum_{j\in\mathcal{C}_{i}}}}f_{j}(\mathbf{x}_{i},\mathbf{x}_{-i})+\boldsymbol{{\pi}}_{\mathcal{C}_{i}}(\mathbf{x}^{\star})^{T}\mathbf{x}_{i}.\label{eq:Game_pricing}
\end{equation}

\end{proposition} 

According to the above proposition, the stationary solutions of (\ref{eq:social problem})
achievable as fixed-points of $\widehat{\mathbf{x}}_{\mathcal{C}_{i}}(\mathbf{\bullet},\boldsymbol{{\tau}})$
are \emph{unilaterally} optimal for the objective functions in (\ref{eq:Game_pricing}).
This result is in agreement with those obtained in \cite{SchmidtShiBerryHonigUtschick-SPMag,Wang-Krunz-Cui_JSTSP08}
for the sum-rate maximization problem over SISO frequency selective-channels.
Despite its theoretical interest, however, Prop. \ref{Lemma_NE} does
not help in practice to solve (\ref{eq:social problem}). Indeed,
the computation of a Nash equilibrium of the game in (\ref{eq:Game_pricing})
would require the a-priori knowledge of the prices $\boldsymbol{{\pi}}_{\mathcal{C}_{i}}(\mathbf{x}^{\star})$
and thus the equilibrium itself, which of course is not available.

\section{Distributed Decomposition Algorithms \label{sub:Jacobi-Distributed-Pricing}}

We are now ready to introduce our new algorithms, as a direct product
of Prop. \ref{Prop_x_y}. We first focus on (inexact) Jacobi schemes
(cf. Sec. \ref{sub:Jacobi-based-schemes}); then we show that the
same results hold also for (inexact) Gauss-Seidel updates (cf. Sec.
\ref{sub:Gauss-Seidel-Implementation}). \vspace{-0.2cm}

\subsection{Exact Jacobi best-response schemes\label{sub:Jacobi-based-schemes}}

The first algorithm we propose is a Jacobi scheme where all users
update simultaneously their strategies based on the best-response
$\widehat{\mathbf{x}}_{\mathcal{C}_{i}}(\mathbf{\bullet},\boldsymbol{{\tau}})$
(possibly with a memory); the formal description is given in Algorithm
\ref{alg:PJA} below, and its convergence properties are given in
Theorem \ref{Theorem_convergence_Jacobi}. \vspace{-0.2cm}

\begin{algorithm}[h]$\textbf{Data}:$ $\boldsymbol{{\tau}}\geq\mathbf{0}$,
$\{\gamma^{n}\}>0$, $\mathbf{x}^{0}\in\mathcal{K}$. Set $n=0$.

\texttt{$\mbox{(S.1)}:$}$\,\,$If $\mathbf{x}^{n}$ satisfies a termination
criterion: STOP;

\texttt{$\mbox{(S.2)}:$} For all $i\in\mathcal{I}$, compute $\widehat{\mathbf{x}}_{\mathcal{C}_{i}}\left(\mathbf{x}^{n},\boldsymbol{{\tau}}\right)$
{[}cf. (\ref{eq:decoupled_problem_i}){]}; 

\texttt{$\mbox{(S.3)}:$} Set $\mathbf{x}^{n+1}\triangleq\mathbf{x}^{n}+\gamma^{n}\,(\widehat{\mathbf{x}}_{\mathcal{C}}\left(\mathbf{x}^{n},\boldsymbol{{\tau}}\right)-\mathbf{x}^{n})$;

\texttt{$\mbox{(S.4)}:$} $n\leftarrow n+1$, and go to \texttt{$\mbox{(S.1)}.$}
\caption{\textbf{:\label{alg:PJA} Exact Jacobi  SCA Algorithm}} \end{algorithm}\vspace{-0.2cm}

\begin{theorem} \label{Theorem_convergence_Jacobi}Given the social
problem\emph{ }(\ref{eq:social problem}) under A1-A4, suppose that
one of the two following conditions is satisfied:

\noindent (a) For each $i$, $\mathbf{H}_{i}(\mathbf{x})$ is such
that $\mathbf{H}_{i}(\mathbf{x})-c_{H_{i}}\mathbf{I}\succeq\mathbf{0}$
for all $\mathbf{x}\in\mathcal{K}$ and some $c_{H_{i}}>0$; furthermore
$\{\gamma^{n}\}$ and $\boldsymbol{{\tau}}\,{\color{blue}\geq}\,\mathbf{0}$
are chosen so that\vspace{-0.5cm}

\begin{equation}
0<\inf_{n}\gamma^{n}\leq\sup_{n}\gamma^{n}\leq\gamma^{\max}\leq1\,\,\mbox{and}\,\,2\, c_{\boldsymbol{{\tau}}}\geq\gamma^{\max}L_{\nabla U},\vspace{-0.1cm}\label{eq:constant_step-size}
\end{equation}
with $c_{\boldsymbol{{\tau}}}$ defined in (\ref{eq:c_tau}). 

\noindent (b)\emph{ }For each $i$, $\mathbf{H}_{i}(\mathbf{x})$
is such that $\mathbf{H}_{i}(\mathbf{x})-c_{H_{i}}\mathbf{I}\succeq\mathbf{0}$
for all $\mathbf{x}\in\mathcal{K}$ and some $c_{H_{i}}>0$, \textcolor{black}{$\boldsymbol{{\tau}}\,\geq\,\mathbf{0}$
is such that $c_{\boldsymbol{{\tau}}}>\mathbf{0}$,} and furthermore\emph{
}$\{\gamma^{n}\}$ is chosen so that\vspace{-0.4cm}

\begin{equation}
\gamma^{n}\in(0,1],\quad\gamma^{n}\rightarrow0,\quad\mbox{and}\quad\sum_{n}\gamma^{n}=+\infty.\vspace{-0.2cm}\label{eq:diminishing_step_size}
\end{equation}
Then, either Algorithm\emph{ }\ref{alg:PJA}\emph{ }converges in a
finite number of iterations to a stationary solution of\emph{ }(\ref{eq:social problem})\emph{
}or every limit point of the sequence\emph{ $\{\mathbf{x}^{n}\}_{n=1}^{\infty}$
}(at least one such point exists) is a stationary solution of (\ref{eq:social problem})\emph{.}
Moreover, none of such points is a local maximum of $U$.\end{theorem}\begin{proof}See
Appendix B. \end{proof}

\noindent \textbf{Main features of Algorithm \ref{alg:PJA}}\emph{.}
The algorithm implements a novel \emph{distributed} SCA decomposition:
all the users solve \emph{in parallel} a sequence of \emph{decoupled}
strongly convex optimization problems as in (\ref{eq:decoupled_problem_i}).
The algorithm is expected to perform better than classical gradient-based
schemes (at least in terms of convergence speed) at the cost of no
extra signaling, because the structure of the objective functions
is better preserved. It is guaranteed to converge under very mild
assumptions (the weakest available in the literature) while offering
some flexibility in the choice of the free parameters {[}conditions
(a) or (b) of Theorem \ref{Theorem_convergence_Jacobi}{]}. This degree
of freedom can be exploited, e.g., to achieve the desired tradeoff
between signaling, convergence speed, and computational effort, as
discussed next.

As far as the computation of the best-response $\widehat{\mathbf{x}}_{\mathcal{C}_{i}}\left(\mathbf{x}^{n},\boldsymbol{{\tau}}\right)$
is concerned, at each iteration, every user needs to known $\sum_{j\in\mathcal{C}_{i}}f_{j}(\bullet,\mathbf{x}_{-i}^{n})$
and $\boldsymbol{{\pi}}_{\mathcal{C}_{i}}(\mathbf{x}^{n})$. The signaling
required to acquire this information is of course problem-dependent.
If the problem under consideration does not have any specific structure,
the most natural message-passing strategy is to communicate directly
$\mathbf{x}_{-i}^{n}$ and $(\nabla_{\mathbf{x}_{i}}f_{j}(\mathbf{x}^{n}))_{j\notin\mathcal{C}_{i}}$.
However, in many specific applications much less signaling may be
needed; see Sec. \ref{sec:Applications} for some examples. 

\noindent\textbf{On the choice of the free parameters.} Convergence
of Algorithm \ref{alg:PJA} is guaranteed either using a constant
step-size rule {[}cf. (\ref{eq:constant_step-size}){]} or a diminishing
step-size rule {[}cf. (\ref{eq:diminishing_step_size}){]}. Moreover,
different choices of $\{\mathcal{C}_{i}\}$ are in general feasible
for a given social function, resulting in different best-response
functions and signaling among the users.

\subsubsection{Constant step-size}

In this case, $\gamma^{n}=\gamma\leq\gamma^{\max}$ for all $n$,
where $\gamma^{\max}\in(0,1]$ needs to be chosen together with $\boldsymbol{{\tau}}\geq\mathbf{0}$
and $(\mathbf{H}_{i}(\mathbf{y}))_{i=1}^{I}$ so that the condition
$2\, c_{\boldsymbol{{\tau}}}\geq\gamma^{\max}L_{\nabla U}$ is satisfied,
with $c_{\boldsymbol{{\tau}}}$ defined in (\ref{eq:c_tau}). This
can be done in several ways.\emph{ }A simple (but conservative) choice
satisfying that condition is, e.g., $\tau_{i}=\tau>0$ for all $i\in\mathcal{I}$,
$\gamma^{\max}\in(0,1]$, and $\gamma/\tau\leq2/L_{\nabla U}$. Note
that this condition imposes a constraint only on the ratio $\gamma/\tau$,
leaving free the choice of one of the two parameters.

An interesting special case worth mentioning is: $\gamma=\gamma^{\max}=1$\emph{
}for all $n$, $\mathbf{H}_{i}(\mathbf{y})=\mathbf{I}$ for all $i\in\mathcal{I}$,
and $\boldsymbol{{\tau}}>\mathbf{0}$ large enough so that $2\, c_{\boldsymbol{{\tau}}}\geq L_{\nabla U}$.
This choice leads to the classical Jacobi best-response scheme (but
with a proximal regularization), namely: at each iteration $n$, 
\[
\mathbf{x}_{i}^{n+1}=\widehat{\mathbf{x}}_{\mathcal{C}_{i}}\left(\mathbf{x}^{n},\boldsymbol{{\tau}}\right),\qquad\forall\in\mathcal{I}.
\]
To the best of our knowledge, this algorithm along with its convergence
conditions {[}Theorem \ref{Theorem_convergence_Jacobi}a){]} represents
a new result in the optimization literature; indeed classic best-response
nonlinear Jacobi schemes require much stronger (sufficient) conditions
to converge (implying contraction) \cite[Ch. 3.3.5]{Bertsekas_Book-Parallel-Comp}.
Note that the choice of $\tau_{i}$'s to guarantee convergence {[}i.e.,
$2\, c_{\boldsymbol{{\tau}}}\geq L_{\nabla U}${]} can be done locally
by each user with no signaling exchange, once the Lipschitz constant
$L_{\nabla U}$ is known.

As a final remark, we point out that in the case of constant and ``sufficiently''
small step-size $\gamma^{n}$, one can relax the synchronization requirements
among the users allowing (partially) asynchronous updates of users
best-responses (in the sense of \cite{Bertsekas_Book-Parallel-Comp});
we omit the details because of space limitation.

\subsubsection{Variable step-size}

In scenarios where the knowledge of the system parameters, e.g. $L_{\nabla U}$,
is not available, one can use the diminishing step-size rule (\ref{eq:diminishing_step_size}).
\textcolor{black}{Under such a rule, convergence is guaranteed for
}\textcolor{black}{\emph{any}}\textcolor{black}{{} choice of $\mathbf{H}_{i}(\mathbf{x})-c_{H_{i}}\mathbf{I}\succeq\mathbf{0}$
and $\boldsymbol{{\tau}}\geq\mathbf{0}$ such that $c_{\boldsymbol{{\tau}}}>0$.
Note that if $\sum_{j\in\mathcal{C}_{i}}f_{j}(\bullet,\mathbf{x}_{-i})$
is strongly convex on $\mathcal{K}_{i}$ for any $\mathbf{x}_{-i}\in\mathcal{K}_{-i}$,
one can also set $\tau_{i}=0$, otherwise any arbitrary but positive
$\tau_{i}$ is necessary. }We will show in the next section that a
diminishing step-size rule is also useful to allow an inexact computation
of the best-response $\widehat{\mathbf{x}}_{\mathcal{C}_{i}}\left(\mathbf{x}^{n},\boldsymbol{{\tau}}\right)$
while preserving convergence of the algorithm. Two classes of step-size
rules satisfying (\ref{eq:diminishing_step_size}) are: given $\gamma^{0}=1$,
\begin{eqnarray}
\mbox{Rule\#1:\,\,\,} & \gamma^{n} & \!\!\!\!=\,\gamma^{n-1}\left(1-\epsilon\,\gamma^{n-1}\right),\quad n=1,\ldots,\medskip\label{eq:step-size_1}\\
\mbox{Rule\#2:\,\,\,} & \gamma^{n} & \!\!\!\!=\,\dfrac{{\gamma^{n-1}+\alpha(n)}}{1+\beta(n)},\quad\quad\quad\, n=1,\ldots,\label{eq:step-size_2}
\end{eqnarray}
where in (\ref{eq:step-size_1}) $\epsilon\in(0,1)$ is a given constant,
whereas in (\ref{eq:step-size_2}) $\alpha(n)$ and $\beta(n)$ are
two nonnegative real functions of $n\geq1$ such that: i) $0\leq\alpha(n)\le\beta(n)$;
and ii) $ $$\alpha(n)/\beta(n)\rightarrow0$ as $n\rightarrow\infty$
while $\sum_{n}\left(\alpha(n)/\beta(n)\right)=\infty$. Examples
of such $\alpha(n)$ and $\beta(n)$ are: $\alpha(n)=\alpha$ or $\alpha(n)=\log(n)^{\alpha}$,
and $\beta(n)=\beta\, n$ or $\beta(n)=\beta\,\sqrt{{n}}$, where
$\alpha,\beta$ are given constants satisfying $\alpha\in(0,1)$,
$\beta\in(0,1)$, and $\alpha\leq\beta$. 

Another issue to discuss is the choice of the free positive definite
matrices $\mathbf{H}_{i}(\mathbf{y})$. Mimicking (quasi-)Newton-like
schemes \cite{Bertsekas_NLPbook99}, a possible choice is to consider
for $\mathbf{H}_{i}(\mathbf{x}^{n})$ a proper (diagonal) uniformly
positive definite ``approximation'' of the Hessian matrix $\nabla_{\mathbf{x}_{i}}^{2}U(\mathbf{x}^{n})$.
The exact expression to consider depends on the amount of signaling
and computational complexity required to compute such a $\mathbf{H}_{i}(\mathbf{x}^{n}),$
and thus varies with the specific problem under consideration.

\subsubsection{On the choice of $\mathcal{C}_{i}$'s}

In general, more than one (feasible) choice of $\{\mathcal{C}_{i}\}$
is possible for a given social function, resulting in different decomposition
schemes. Some illustrative examples are discussed next. \smallskip

\noindent\emph{Example \#1$-$(Proximal) gradient/}\textcolor{black}{\emph{Newton}}\emph{
algorithms}: If each $\mathcal{C}_{i}=\emptyset$ and $I=I_{f}$,
$\widehat{\mathbf{x}}_{\mathcal{C}_{i}}(\mathbf{x}^{n},\tau_{i})$
reduces to the gradient response (\ref{eq:lin_prob}) (possibly with
a proximal regularization). It turns out that (exact and inexact)
gradient algorithms along with their convergence conditions are special
cases of our framework. Note that if $\mathcal{S}_{i}=\emptyset$
for every $i$ (i.e., no convexity whatsoever is present in $U$),
this is the only possible choice, and indeed our approach reduces
to a gradient-like method. On the other hand, as soon as at least
some $\mathcal{S}_{i}\neq\emptyset$, we may depart from the gradient
method and exploit the available convexity. 

\textcolor{black}{Note that our framework contains also Newtown-like
updates. For instance, if $U(\mathbf{x}_{i},\mathbf{x}_{-i}^{n})$
is convex in $\mathbf{x}_{i}\in\mathcal{K}_{i}$ for any $\mathbf{x}_{-i}^{n}\in\mathcal{K}_{-i}$,
a feasible choice is $\mathcal{C}_{i}=\emptyset$ and $\mathbf{H}_{i}(\mathbf{x}^{n})=\nabla_{\mathbf{x}_{i}}^{2}U(\mathbf{x}^{n})$,
resulting in: 
\begin{equation}
\begin{array}{l}
\widehat{\mathbf{x}}_{i}(\mathbf{x}^{n},\tau_{i})\triangleq\underset{\mathbf{x}_{i}\in\mathcal{K}_{i}}{\mbox{argmin}\,}\left\{ \nabla_{\mathbf{x}_{i}}U(\mathbf{x}^{n})^{T}(\mathbf{x}_{i}-\mathbf{x}_{i}^{n})\right.\\
\qquad\quad\qquad\qquad\qquad+\dfrac{{1}}{2}\,(\mathbf{x}_{i}-\mathbf{x}_{i}^{n})^{T}\nabla_{\mathbf{x}_{i}}^{2}U(\mathbf{x}^{n})(\mathbf{x}_{i}-\mathbf{x}_{i}^{n})\smallskip\\
\left.\qquad\,\,\,\,\,\qquad\qquad\qquad+\dfrac{{\tau_{i}}}{2}\left\Vert \mathbf{x}_{i}-\mathbf{x}_{i}^{n}\right\Vert ^{2}\right\} .
\end{array}\label{eq:Newton_like}
\end{equation}
Essentially (\ref{eq:Newton_like}) corresponds to a Newton-like step
of user $i$ in minimizing the ``reduced'' problem $\min_{\mathbf{x}_{i}\in\mathcal{K}_{i}}U(\mathbf{x}_{i},\mathbf{x}_{-i}^{n})$.}

\noindent\emph{Example \#}2\emph{$-$Pricing algorithms in }\cite{ScutariPalomarFacchineiPang_NETGCOP11}:
Suppose that $I=I_{f}$, and each $\mathcal{S}_{i}=\{i\}$ (implying
that $f_{i}(\bullet,\mathbf{x}_{-i})$ is convex on $\mathcal{K}_{i}$
for any $\mathbf{x}_{-i}\in\mathcal{K}_{-i}$). By taking each $\mathcal{C}_{i}=\{i\}$
and $\mathbf{H}_{i}(\mathbf{x}^{n})=\mathbf{I}$, we obtain the pricing-based
algorithms in \cite{ScutariPalomarFacchineiPang_NETGCOP11}:
\[
\widehat{\mathbf{x}}_{i}(\mathbf{x}^{n},\tau_{i})\triangleq\underset{\mathbf{x}_{i}\in\mathcal{K}_{i}}{\mbox{argmin}\,}f_{i}(\mathbf{x}_{i},\mathbf{x}_{-i}^{n})+\boldsymbol{{\pi}}_{i}(\mathbf{x}^{n})^{T}\mathbf{x}_{i}+\frac{{\tau_{i}}}{2}\left\Vert \mathbf{x}_{i}-\mathbf{x}_{i}^{n}\right\Vert ^{2},
\]
where $\boldsymbol{\pi}_{i}(\mathbf{x}^{n})\triangleq\sum_{j\neq i}\nabla_{\mathbf{x}_{i}}f_{j}(\mathbf{x}^{n})$.
Algorithm \ref{alg:PJA} based on the above best-response implements
naturally a pricing mechanism; indeed, each $\boldsymbol{\pi}_{i}(\mathbf{x}^{n})$
represents a dynamic pricing that measures somehow the marginal increase
of the sum-utility of the other users due to a variation of the strategy
of user $i$; roughly speaking, it works like a punishment imposed
to each user for being too aggressive in choosing his own strategy
and thus ``hurting'' the other users. Pricing algorithms based on
heuristics have been proposed in a number of papers for the sum-rate
maximization problem over SISO/SIMO/MIMO ICs \cite{Huang-Berry-Honig_JSAC06,Wang-Krunz-Cui_JSTSP08,ShiBerryHonig-CISS08,ShiSchmidtBerryHonigUtschick-ICC09,SchmidtShiBerryHonigUtschick-SPMag}.
However, on top of being \emph{sequential} schemes, convergence of
algorithms in the aforementioned papers is established under relatively
strong assumptions (e.g., limited number of users, special classes
of functions, specific channel models and transmission schemes, etc...),
see \cite{SchmidtShiBerryHonigUtschick-SPMag}. The pricing in our
framework is instead the natural consequence of the proposed SCA decomposition
technique and leads to \emph{simultaneous} algorithms that can be
applied (with convergence guaranteed) to a very large class of problems,
even when \cite{Huang-Berry-Honig_JSAC06,Wang-Krunz-Cui_JSTSP08,ShiBerryHonig-CISS08,ShiSchmidtBerryHonigUtschick-ICC09,SchmidtShiBerryHonigUtschick-SPMag}
fail.\smallskip

\noindent\emph{Example \#3$-$(Proximal) Jacobi algorithms} \emph{for
a single jointly convex function: }Suppose that the social function
is a single (jointly) convex function $f(\mathbf{x}_{1},\ldots,\mathbf{x}_{I})$
on $\mathcal{K}=\prod_{i}\mathcal{K}_{i}$. Of course, this optimization
problem can be interpreted as a special case of the framework (\ref{eq:social problem}),
with $\mathcal{C}_{i}=\mathcal{S}_{i}=\{1\}=\mathcal{I}_{f}$, for
all $i\in\mathcal{I}$ and $f_{1}(\mathbf{x})=f(\mathbf{x})$. Then,
setting $\mathbf{H}_{i}(\mathbf{x}^{n})=\mathbf{I}$, the best-response
(\ref{eq:decoupled_problem_i}) of each user $i$ reduces to 
\begin{equation}
\widehat{\mathbf{x}}_{\mathcal{C}_{i}}(\mathbf{x}^{n},\tau_{i})\triangleq\underset{\mathbf{x}_{i}\in\mathcal{K}_{i}}{\mbox{argmin}\,}f(\mathbf{x}_{i},\mathbf{x}_{-i}^{n})+\frac{{\tau_{i}}}{2}\left\Vert \mathbf{x}_{i}-\mathbf{x}_{i}^{n}\right\Vert ^{2}.\label{eq:Jacobi_single_function}
\end{equation}
 Algorithm \ref{alg:PJA} based on (\ref{eq:Jacobi_single_function})
reads as a block-Jacobi schemes converging to the global minima of
$f(\mathbf{x}_{1},\ldots,\mathbf{x}_{I})$ over $\mathcal{K}$ (cf.
Theorem \ref{Theorem_convergence_Jacobi}). To the best of our knowledge,
these are new algorithms in the literature; moreover their convergence
conditions enlarge current ones; see, e.g., \cite[Sec. 3.2.4]{Bertsekas_Book-Parallel-Comp}.
Quite interestingly, this new algorithm can be readily applied to
solve the sum-rate maximization over MIMO \emph{multiple access} channels
\cite{YuRheeBoydCioffi_IT04}, resulting in the first (inexact)\emph{
simultaneous }MIMO iterative waterfilling algorithm in the literature;
we omit the details because of the space limitation.\smallskip

\noindent\emph{Example \#4$-$Algorithms} \emph{for DC programming.
}The proposed framework applies naturally to sum-utility problems
where the users' functions are the difference of two convex functions,
namely:\vspace{-0.1cm} 
\begin{equation}
\begin{array}{ll}
\underset{\mathbf{x}_{1},\ldots,\mathbf{x}_{I}}{\text{minimize}} & {\displaystyle {\sum_{i\in\mathcal{I}}}}\, f_{i}^{\text{{cvx}}}(\mathbf{x})+{\displaystyle {\sum_{i\in\mathcal{I}}}}\, f_{i}^{\text{{ccv}}}(\mathbf{x})\\
\text{subject to} & \mathbf{x}_{i}\in\mathcal{K}_{i},\,\forall i\in\mathcal{I}
\end{array}\label{eq:DC_programming}
\end{equation}
where $f_{i}^{\text{{cvx}}}(\mathbf{x})$ and $f_{i}^{\text{{ccv}}}(\mathbf{x})$
are convex and concave functions on $\mathcal{K}$, respectively.
Letting 
\[
{\displaystyle f_{1}(\mathbf{x})\triangleq{\sum_{i\in\mathcal{I}}}}\, f_{i}^{\text{{cvx}}}(\mathbf{x})\quad\mbox{and}\quad{\displaystyle f_{2}(\mathbf{x})\triangleq{\sum_{i\in\mathcal{I}}}}\, f_{i}^{\text{{ccv}}}(\mathbf{x}),
\]
the optimization problem (\ref{eq:DC_programming}) can be interpreted
as a special case of the framework (\ref{eq:social problem}), with
$\mathcal{I}_{f}=\{1,2\}$, $\mathcal{C}_{i}=\{1\}$ for all $i\in\mathcal{I}${\small{.}}
The best-response (\ref{eq:decoupled_problem_i}) of each user $i$
reduces then to 
\begin{eqnarray}
\widehat{\mathbf{x}}_{\mathcal{C}_{i}}(\mathbf{x}^{n},\tau_{i}) & = & \underset{\mathbf{x}_{i}\in\mathcal{K}_{i}}{\mbox{argmin}}\left\{ f_{1}(\mathbf{x}_{i},\mathbf{x}_{-i}^{n})+\boldsymbol{{\pi}}_{i}(\mathbf{x}^{n})^{T}\mathbf{x}_{i}\right.\vspace{-0.2cm}\nonumber \\
 &  & \left.\qquad\quad\,\,+\frac{{\tau_{i}}}{2}\left\Vert \mathbf{x}_{i}-\mathbf{x}_{i}^{n}\right\Vert ^{2}\right\} \label{eq:BR_DC}
\end{eqnarray}
where $\boldsymbol{{\pi}}_{i}(\mathbf{x}^{n})\triangleq\nabla_{\mathbf{x}_{i}}f_{2}(\mathbf{x}^{n})$
and $\mathbf{H}_{i}(\mathbf{x}^{n})=\mathbf{I}$. The above decomposition
can be applied, e.g., to the sum-rate maximization (\ref{eq:SISO_formulation}),
when all $\theta_{i}(x)=w_{i}\, x$, with $w_{i}>0$; see Sec. \ref{sec:Applications}.
\vspace{-0.2cm}

\subsection{Inexact Jacobi best-response schemes\label{sub:Inexact-Jacobi-best-response}}

In many practical network settings, it can be useful to further reduce
the computational \textcolor{black}{effort needed to solve} users'
(convex) sub-problems (\ref{eq:decoupled_problem_i}) by allowing
inexact computations of the best-response functions $\widehat{\mathbf{x}}_{\mathcal{C}_{i}}\left(\mathbf{x}^{n},\boldsymbol{{\tau}}\right)$.
Algorithm \ref{alg:PJA-inex} is a variant of Algorithm \ref{alg:PJA},
in which suitable approximations of $\widehat{\mathbf{x}}_{\mathcal{C}_{i}}\left(\mathbf{x}^{n},\boldsymbol{{\tau}}\right)$
can be used. \vspace{-0.2cm}

\begin{algorithm}[h]$\textbf{Data}:$ $\{\varepsilon_{i}^{n}\}$
for $i\in\mathcal{I}$, $\boldsymbol{{\tau}}\geq\mathbf{0}$, $\{\gamma^{n}\}>0$,
$\mathbf{x}^{0}\in\mathcal{K}$. Set $n=0$.

\texttt{$\mbox{(S.1)}:$}$\,\,$If $\mathbf{x}^{n}$ satisfies a termination
criterion: STOP;

\texttt{$\mbox{(S.2)}:$} For all $i\in\mathcal{I}$, solve (\ref{eq:decoupled_problem_i})
within the accuracy $\varepsilon_{i}^{n}:$ 

\hspace{1.36cm}Find $\mathbf{z}_{i}^{n}$ s.t. $\|\mathbf{z}_{i}^{n}-\widehat{\mathbf{x}}_{\mathcal{C}_{i}}\left(\mathbf{x}^{n},\boldsymbol{{\tau}}\right)\|\leq\varepsilon_{i}^{n}$; 

\texttt{$\mbox{(S.3)}:$} Set $\mathbf{x}^{n+1}\triangleq\mathbf{x}^{n}+\gamma^{n}\,(\mathbf{z}^{n}-\mathbf{x}^{n})$;

\texttt{$\mbox{(S.4)}:$} $n\leftarrow n+1$, and go to \texttt{$\mbox{(S.1)}.$}
\caption{\textbf{:\label{alg:PJA-inex} Inexact Jacobi  SCA Algorithm}} \end{algorithm}\vspace{-0.2cm}

The error term $\varepsilon_{i}^{n}$ in Step 2 measures the accuracy
used at iteration $n$ in computing the solution $\widehat{\mathbf{x}}_{\mathcal{C}_{i}}\left(\mathbf{x}^{n},\boldsymbol{{\tau}}\right)$
of each problem (\ref{eq:decoupled_problem_i}). Note that if we set
$\varepsilon_{i}^{n}=0$ for all $n$ and $i$, Algorithm \ref{alg:PJA-inex}
reduces to Algorithm \ref{alg:PJA}. Obviously, the errors $\varepsilon_{i}^{n}$\textquoteright{}s
and the step-size $\gamma^{n}$\textquoteright{}s must be chosen according
to some suitable conditions, if one wants to guarantee convergence.
These conditions are established in the following theorem.

\begin{theorem} \label{Theorem_convergence_inexact_Jacobi}Let\emph{
$\{\mathbf{x}^{n}\}_{n=1}^{\infty}$} be the sequence generated by
Algorithm \ref{alg:PJA-inex}, under the setting of Theorem \ref{Theorem_convergence_Jacobi}
where however we reenforce assumption A4 by assuming that $U$ is
coercive on $\mathcal{K}$. Suppose that the sequences\emph{ }$\{\gamma^{n}\}$
and $\{\varepsilon_{i}^{n}\}$ satisfy the following conditions: i)
$\gamma^{n}\in(0,1]$; ii) $\gamma^{n}\rightarrow0$; iii) $\sum_{n}\gamma^{n}=+\infty$;
iv) $\sum_{n}\left(\gamma^{n}\right)^{2}<+\infty$; and v) $\sum_{n}\varepsilon_{i}^{n}\gamma^{n}<+\infty$
for all $i=1,\ldots,I$. Then, either Algorithm\emph{ }\ref{alg:PJA-inex}\emph{
}converges in a finite number of iterations to a stationary solution
of\emph{ }(\ref{eq:social problem})\emph{ }or every limit point of
the sequence\emph{ $\{\mathbf{x}^{n}\}_{n=1}^{\infty}$ }(at least
one such points exists) is a stationary solution of (\ref{eq:social problem})\emph{.}
\end{theorem}\begin{proof}See Appendix B. \end{proof}

As expected, in the presence of errors, convergence of Algorithm \ref{alg:PJA-inex}
is guaranteed if the sequence of approximated problems (\ref{eq:decoupled_problem_i})
is solved with increasing accuracy. Note that, in addition to requiring
$\varepsilon_{i}^{n}\rightarrow0$, condition v) of Theorem \ref{Theorem_convergence_inexact_Jacobi}
imposes also a constraint on the rate by which the $\varepsilon_{i}^{n}$
go to zero, which depends on the rate of decrease of $\{\gamma^{n}\}$.
Two instances of step-size rules satisfying the summability condition
iv) are given by (\ref{eq:step-size_1}) and (some choices of) (\ref{eq:step-size_2}).
An example of error sequence satisfying condition v) is $\varepsilon_{i}^{n}\leq c_{i}\,\gamma^{n}$,
where $c_{i}$ is any finite positive constant. Such a condition can
be forced in Algorithm \ref{alg:PJA-inex} in a distributed way, using
classical error bound results in convex analysis; see, e.g., \cite[Ch. 6, Prop. 6.3.7]{Facchinei-Pang_FVI03}.

\textcolor{black}{Finally, it is worth observing that }Algorithm \ref{alg:PJA-inex}\textcolor{black}{{}
}(and \ref{alg:PJA}\textcolor{black}{) with a diminishing step-size
rule satisfying i)-iv) of }Theorem \ref{Theorem_convergence_inexact_Jacobi}\textcolor{black}{{}
can be made robust against (stochastic) errors on the price estimates,
due to an imperfect communication scenario (random link failures,
noisy estimate, quantization, etc...). Because of the space limitation,
we do not further elaborate on this here; see \cite{YangScutariPalomarSPAWC13}
for details. }\vspace{-0.2cm}

\subsection{(Inexact) Gauss-Seidel best-response schemes \label{sub:Gauss-Seidel-Implementation}\vspace{-0.1cm}}

The Gauss-Seidel implementation of the proposed SCA decomposition
is described in Algorithm \ref{alg:PGSA}, where the users solve sequentially,
in an exact or inexact form, the convex subproblems (\ref{eq:decoupled_problem_i}).
In the algorithm, we used the notation $\mathbf{x}_{i<}^{t+1}\triangleq(\mathbf{x}_{1}^{t+1},\ldots,\mathbf{x}_{i-1}^{t+1})$
and $\mathbf{x}_{i\geq}^{t}\triangleq(\mathbf{x}_{i}^{t},\ldots,\mathbf{x}_{I}^{t})$.

\begin{algorithm}[t] $\textbf{Data}:$ $\{\varepsilon_{i}^{t}\}$
for $i\in\mathcal{I}$, $\boldsymbol{{\tau}}\geq\mathbf{0}$, $\{\gamma^{t}\}>0$,
$\mathbf{x}^{0}\in\mathcal{K}$. Set $t=0$.

\texttt{$\mbox{(S.1)}:$}$\,\,$If $\mathbf{x}^{t}$ satisfies a termination
criterion: STOP;

\texttt{$\mbox{(S.2)}:$} For $i=1,\ldots,I$, 

\hspace{1.5cm}a) Find $\mathbf{z}_{i}^{t}$ s.t. $\|\mathbf{z}_{i}^{t}-\widehat{\mathbf{x}}_{\mathcal{C}_{i}}\left((\mathbf{x}_{i<}^{t+1},\mathbf{x}_{i\geq}^{t}),\boldsymbol{{\tau}}\right)\|\leq\varepsilon_{i}^{t}$;

\hspace{1.5cm}b) Set $\mathbf{x}_{i}^{t+1}\triangleq\mathbf{x}_{i}^{t}+{\gamma}^{\, t}\,\left(\mathbf{z}_{i}^{t}-\mathbf{x}_{i}^{t}\right)$

\texttt{$\mbox{(S.3)}:$} $t\leftarrow t+1$, and go to \texttt{$\mbox{(S.1)}.$}\caption{\textbf{:\label{alg:PGSA} Inexact Gauss-Seidel SCA Algorithm}} \end{algorithm}

Note that one round of Algorithm \ref{alg:PGSA} (i.e., $t\leftarrow t+1$)
wherein all users sequentially update their own strategies, corresponds
to $I$ consecutive iterations $n$ of the Jacobi updates described
in Algorithms \ref{alg:PJA} and \ref{alg:PJA-inex}. In Appendix
C we prove that, quite interestingly, Algorithm \ref{alg:PGSA} can
be interpreted as an inexact Jacobi scheme based on the best-response
$\widehat{\mathbf{x}}_{\mathcal{C}}\left(\bullet,\boldsymbol{{\tau}}\right)$,
satisfying Theorem \ref{Theorem_convergence_inexact_Jacobi}. It turns
out that convergence of Algorithm \ref{alg:PGSA} follows readily
from that of Algorithm \ref{alg:PJA-inex}, and is stated next. \smallskip

\begin{theorem} \label{Theorem_convergence_GS} Let\emph{ $\{\mathbf{x}^{n}\}_{n=1}^{\infty}$}
be the sequence generated by Algorithm\emph{ }\ref{alg:PGSA}, under
the setting of Theorem \ref{Theorem_convergence_inexact_Jacobi}.
Then, the conclusions of Theorem \ref{Theorem_convergence_inexact_Jacobi}
holds\emph{.}\end{theorem}\begin{proof}See Appendix C. \end{proof}\vspace{-0.2cm}

\section{The Complex Case\label{sec:The-Complex-Case}}

In this section we show how to extend our framework to sum-utility
problems where the users' optimization variables are complex matrices.
This will allow us to deal with the design of MIMO multiuser systems.
Let us consider the following sum-utility optimization:\vspace{-.1cm}
\begin{equation}
\begin{array}{ll}
\underset{\mathbf{X}_{1},\ldots,\mathbf{X}_{I}}{\mbox{minimize}} & {\displaystyle U(\mathbf{X})\triangleq{\sum_{\ell\in\mathcal{I}_{f}}}}\, f_{\ell}(\mathbf{X})\\[5pt]
\text{subject to} & \mathbf{X}_{i}\in{\cal X}_{i},\quad\forall i\in\mathcal{I},
\end{array}\vspace{-.1cm}\label{eq:MIMO_sum-utility}
\end{equation}
where $\mathbf{X}\triangleq(\mathbf{X}_{i})_{i\in\mathcal{I}}$, with
$\mathbf{X}_{i}\in\mathbb{C}^{n_{i}\times m_{i}}$ being the (matrix)
strategy of user $i$, $\mathcal{X}_{i}\subseteq\mathbb{C}^{n_{i}\times m_{i}}$,
and $f_{\ell}:\mathcal{X}\rightarrow\mathbb{R}$, with $\mathcal{X}\triangleq\prod_{i\in\mathcal{I}}\mathcal{X}_{i}$;
let define also $\mathcal{X}_{-i}\triangleq\prod_{j\neq i}\mathcal{X}_{j}$.
We study (\ref{eq:MIMO_sum-utility}) under the same assumptions A1-A4
stated for the real case, where in A2 the differentiability condition
is now replaced by the $\mathbb{R}$-differentiability (see, e.g.,
\cite{Are_book_MatrixDiff,Scutari-Facchinei-Pang-Palomar_IT_PI}),
and in A3 $U(\mathbf{X})$ is required to have Lipschitz \emph{conjugate-gradient}
$\nabla_{\mathbf{X}^{\ast}}U(\mathbf{X})$ on $\mathcal{K}$, with
constant $L_{\nabla U}^{\mathbb{\mathbb{C}}}$\emph{, }where\emph{
$\mathbf{X}^{\ast}$ }is the conjugate of $\mathbf{X}$.\smallskip

\noindent \emph{A motivating example.} An instance of (\ref{eq:MIMO_sum-utility})
is the MIMO version of (\ref{eq:SISO_formulation}): \vspace{-0.3cm}

\begin{equation}
\begin{array}{ll}
\underset{\mathbf{Q}_{1},\ldots,\mathbf{Q}_{I}}{\mbox{maximize}} & {\displaystyle {\sum_{i\in\mathcal{I}}}}\,\theta_{i}\left(R_{i}(\mathbf{Q}_{i},\mathbf{Q}_{-i})\right)\medskip\\[5pt]
\text{subject to} & \mathbf{Q}_{i}\in\mathcal{Q}_{i},\quad\forall i\in\mathcal{I}.
\end{array}\label{eq:MIMO_formulation_rate}
\end{equation}
where $R_{i}(\mathbf{Q}_{i},\mathbf{Q}_{-i})$ is the rate over the
MIMO link $i$, 
\begin{equation}
R_{i}(\mathbf{Q}_{i},\mathbf{Q}_{-i})\triangleq\log\det\left(\mathbf{I}+\mathbf{H}_{ii}^{H}\mathbf{R}_{i}(\mathbf{Q}_{-i})^{-1}\mathbf{H}_{ii}\mathbf{Q}_{i}\right),\label{eq:rate_MIMO}
\end{equation}
$\mathbf{Q}_{i}$ is the covariance matrix of transmitter $i$, $\mathbf{R}_{i}(\mathbf{Q}_{-i})\triangleq\mathbf{R}_{n_{i}}+\sum_{j\neq i}\mathbf{H}_{ij}\mathbf{Q}_{j}\mathbf{H}_{ij}^{H}$
is the covariance matrix of the multiuser interference plus the thermal
noise $\mathbf{R}_{n_{i}}$ (assumed to be full-rank), with $\mathbf{Q}_{-i}\triangleq(\mathbf{Q}_{j})_{j\neq i}$,
$\mathbf{H}_{ij}$ is the channel matrix between the $j$-th transmitter
and the $i$-th receiver, and $\mathcal{Q}_{i}$ is the set of constraints
of user $i$, 
\[
\mathcal{Q}_{i}\triangleq\left\{ \mathbf{Q}_{i}\in\mathbb{C}^{n_{i}\times n_{i}}:\mathbf{Q}_{i}\succeq\mathbf{0},\,\text{{tr}}(\mathbf{Q}_{i})\leq P_{i},\,\mathbf{Q}_{i}\in\mathcal{Z}_{i}\right\} .
\]
In $\mathcal{Q}_{i}$ we also included an arbitrary convex and closed
set $\mathcal{Z}_{i}$, which allows us to add additional constraints,
such as: i) null constraints $\mathbf{U}_{i}^{H}\mathbf{Q}_{i}\!=\!\mathbf{0}$,
where $\mathbf{U}_{i}\!\in\!\mathbb{C}^{n_{i}\times r_{i}}$ is a
full rank matrix with $r_{i}<n_{i}$; ii) soft-shaping constraints
$\mathsf{tr}\left(\mathbf{G}_{i}^{H}\mathbf{Q}_{i}\mathbf{G}_{i}\right)\!\leq\! I_{i}^{\text{{ave}}}$,
with $\mathbf{G}_{i}\!\in\!\mathbb{C}^{n_{i}\times m_{G_{i}}}$ for
some $m_{G_{i}}>0$; iii) peak-power constraints $\lambda_{\max}\left(\mathbf{F}_{i}^{H}\mathbf{Q}_{i}\mathbf{F}_{i}\right)\!\leq\! I_{i}^{\text{{peak}}}$,
with $\mathbf{F}_{i}\!\in\!\mathbb{C}^{n_{i}\times m_{F_{i}}}$ for
some $m_{F_{i}}>0$; and iv) per-antenna constraints $[\mathbf{Q}_{i}]_{kk}\leq\alpha_{ik}$.
Note that the optimization problems in \cite{Ye-Blum_SP03,KimGiannakisIT11,SchmidtShiBerryHonigUtschick-SPMag}
are special cases of (\ref{eq:MIMO_formulation_rate}). \vspace{-0.2cm}

\subsection{Distributed decomposition algorithms\label{sub:Distributed-decomposition-algorithm_complex}}

At the basis of the proposed decomposition techniques for (\ref{eq:MIMO_sum-utility})
there is the (second order) Taylor expansion of a continuously $\mathbb{R}$-differentiable
function $f:\mathbb{C}^{n\times m}\rightarrow\mathbb{R}$ \cite{Scutari-Facchinei-Pang-Palomar_IT_PI}:
\begin{equation}
\!\!\!\!\begin{array}{l}
f(\mathbf{X}+\Delta\mathbf{X})-f(\mathbf{X})\approx2\left\langle \Delta\mathbf{X},\,\nabla_{\mathbf{X}^{\ast}}f(\mathbf{X})\right\rangle \\
\quad+\dfrac{{1}}{2}\,\text{{vec}}([\mathbf{\Delta\mathbf{X},\Delta\mathbf{X}^{\ast}}])^{H}\mathcal{H}_{\mathbf{X}\mathbf{X}^{\ast}}f(\mathbf{X})\,\text{{vec}}([\mathbf{\mathbf{\Delta\mathbf{X}},\Delta\mathbf{X}^{\ast}}]),
\end{array}\label{eq:Taylor}
\end{equation}
where $\left\langle \mathbf{A},\,\mathbf{B}\right\rangle \triangleq\mbox{Re}\left\{ \text{{tr}}(\mathbf{A}^{H}\mathbf{B})\right\} $,
$\text{{vec}}(\bullet)$ denotes the ``vec'' operator, and $\mathcal{H}_{\mathbf{X}\mathbf{X}^{\ast}}f(\mathbf{X})$
is the so-called \emph{augmented Hessian} of $f$, defined as \cite{Scutari-Facchinei-Pang-Palomar_IT_PI}\vspace{-0.1cm}
\begin{equation}
\mathcal{H}_{\mathbf{X}\mathbf{X}^{\ast}}f(\mathbf{X})\triangleq\dfrac{{\partial}}{\partial\text{{vec}}([\mathbf{\mathbf{X},\mathbf{X}^{\ast}}])^{T}}\left(\dfrac{{\partial}f(\mathbf{X})}{\partial\text{{vec}}([\mathbf{\mathbf{X}^{\ast},\mathbf{X}}])^{T}}\right)^{T}.\label{eq:aug_Hessian}
\end{equation}
In \cite{Scutari-Facchinei-Pang-Palomar_IT_PI}, we proved that $\mathcal{H}_{\mathbf{X}\mathbf{X}^{\ast}}f(\mathbf{X})$
plays the role of the Hessian matrix for functions of real variables.
In particular, $f$ is strongly convex on $\mathbb{C}^{n\times m}$
if and only if there exists a $c_{f^{\mathbb{C}}}>0$, the constant
of strong convexity of $f$, such that 
\begin{equation}
\text{{vec}}([\mathbf{\mathbf{Y},Y^{\ast}}])^{H}\mathcal{H}_{\mathbf{X}\mathbf{X}^{\ast}}f(\mathbf{X})\,\text{{vec}}([\mathbf{\mathbf{Y},Y^{\ast}}])\geq c_{f^{\mathbb{C}}}\left\Vert \mathbf{Y}\right\Vert _{F}^{2},\label{eq:strongly_monotonicity_complex_case}
\end{equation}
for all $\mathbf{X}\in\mathbb{C}^{n\times m}$ and $\mathbf{Y}\in\mathbb{C}^{n\times m}$,
where $\left\Vert \bullet\right\Vert _{F}$ denotes the Frobenius
norm. When (\ref{eq:strongly_monotonicity_complex_case}) holds, we
say that $ $ $\mathcal{H}_{\mathbf{X}\mathbf{X}^{\ast}}f(\mathbf{X})$
is \emph{augmented} uniformly positive definite, and write $\mathcal{H}_{\mathbf{X}\mathbf{X}^{\ast}}f(\mathbf{X})-c_{f^{\mathbb{C}}}\mathbf{I}\overset{\mathcal{A}}{\succeq}\mathbf{0}$
\cite{Scutari-Facchinei-Pang-Palomar_IT_PI}. If $f$ is only convex
but not strongly convex, then $c_{f^{\mathbb{C}}}$ in (\ref{eq:strongly_monotonicity_complex_case})
is zero. 

Motivated by the Taylor expansion (\ref{eq:Taylor}), and using the
same symbols $\mathcal{S}_{i}$ and $\mathcal{C}_{i}$ to denote the
complex counterparts of $\mathcal{S}_{i}$ and $\mathcal{C}_{i}$
introduced for the real case {[}cf. (\ref{eq:C)i_set-1}){]}, let
us consider for each user $i$ the following convex approximation
of $U(\mathbf{X})$ at $\mathbf{X}^{n}$: denoting by $\Delta\mathbf{X}_{i}\triangleq\mathbf{X}_{i}-\mathbf{X}_{i}^{n}$,
\begin{equation}
\!\!\!\!\begin{array}{l}
\widetilde{f}_{\mathcal{C}_{i}}(\mathbf{X}_{i};\mathbf{X}^{n})\triangleq{\displaystyle {\sum_{j\in\mathcal{C}_{i}}}}f_{j}(\mathbf{X}_{i},\,\mathbf{X}_{-i}^{n})+\left\langle \mathbf{\Pi}_{\mathcal{C}_{i}}(\mathbf{X}^{n}),\,\Delta\mathbf{X}_{i}\right\rangle \\
\,\,\,\,\,\,\,+\dfrac{{\tau}_{i}}{2}\,\text{{vec}}([\Delta\mathbf{X}_{i},\Delta\mathbf{X}_{i}^{\ast}])^{H}\mathcal{H}_{i}(\mathbf{X}^{n})\,\text{{vec}}([\Delta\mathbf{X}_{i},\Delta\mathbf{X}_{i}^{\ast}])
\end{array}\label{eq:convex_approx_of_fi_on_Ci_convex}
\end{equation}
with\vspace{-0.1cm} 
\begin{equation}
\mathbf{\Pi}_{\mathcal{C}_{i}}(\mathbf{X}^{n})\triangleq\sum_{j\in\mathcal{C}_{-i}}\!\left.\nabla_{\mathbf{X}_{i}^{\ast}}f_{j}(\mathbf{X})\right|_{\mathbf{X}=\mathbf{X}^{n}},\vspace{-0.1cm}\label{eq:pricing_pi_C_convex}
\end{equation}
where $\mathcal{H}_{i}(\mathbf{X}^{n})$ is any given $2nm\times2nm$
matrix such that $\mathcal{H}_{i}(\mathbf{X})-c_{\mathcal{H}_{i}}\mathbf{I}\overset{\mathcal{A}}{\succeq}\mathbf{0}$,
for all $\mathbf{X}\in\mathcal{X}$ and some $c_{\mathcal{H}_{i}}>0$.
Note that if $\mathcal{H}_{i}(\mathbf{X})=\mathbf{I}$, the quadratic
term in (\ref{eq:convex_approx_of_fi_on_Ci_convex}) reduces to the
standard proximal regularization $\tau_{i}\left\Vert \mathbf{X}_{i}-\mathbf{X}_{i}^{n}\right\Vert _{F}^{2}$.
Then, the best-response matrix function of each user is 
\begin{equation}
\widehat{\mathbf{X}}_{\mathcal{C}_{i}}(\mathbf{X}^{n},\tau_{i})\triangleq\underset{\mathbf{X}_{i}\in\mathcal{X}_{i}}{\mbox{argmin}\,}\widetilde{f}_{\mathcal{C}_{i}}(\mathbf{X}_{i};\mathbf{X}^{n}).\label{eq:best-reponse_c_i_convex}
\end{equation}

Decomposition algorithms for (\ref{eq:MIMO_sum-utility}) are formally
the same as those proposed in Sec. \ref{sub:Jacobi-Distributed-Pricing}
for (\ref{eq:social problem}) {[}namely Algorithms \ref{alg:PJA}-\ref{alg:PGSA}{]},
where the real-valued best-response map $\widehat{\mathbf{x}}_{\mathcal{C}}(\mathbf{x}^{n},\boldsymbol{{\tau}})$
is replaced with the complex-valued counterpart $\widehat{\mathbf{X}}_{\mathcal{C}}(\mathbf{X}^{n},\boldsymbol{{\tau}})\triangleq(\widehat{\mathbf{X}}_{\mathcal{C}_{i}}(\mathbf{X}^{n},\tau_{i}))_{i=1}^{I}$.
Convergence conditions read as in Theorems \ref{Theorem_convergence_Jacobi}-\ref{Theorem_convergence_GS},
under the following natural changes: i) $L_{\nabla U}$ becomes $L_{\nabla U}^{\mathbb{\mathbb{C}}}$;
ii) the condition $\mathbf{H}_{i}(\mathbf{x})-c_{H_{i}}\mathbf{I}\succeq\mathbf{0}$
for all $\mathbf{x}\in\mathcal{K}$ reads as $\mathcal{H}_{i}(\mathbf{X})-c_{\mathcal{H}_{i}}\mathbf{I}\overset{\mathcal{A}}{\succeq}\mathbf{0}$,
for all $\mathbf{X}\in\mathcal{X}$; and iii) in the constant $c_{\boldsymbol{{\tau}}}$
defined in (\ref{eq:c_tau}) $c_{\tau_{i}}\left(\mathbf{x}\right)$
is replaced with $c_{\tau_{i}}(\mathbf{X})$, \textcolor{black}{where
$c_{\tau_{i}}(\mathbf{X})$ is the constant of strong convexity of
$\widetilde{f}_{\mathcal{C}_{i}}(\bullet;\mathbf{X})$ \cite{Scutari-Facchinei-Pang-Palomar_IT_PI}:}\vspace{-0.1cm}\textcolor{blue}{{}
}\textcolor{black}{
\[
\begin{array}{l}
\left\langle \mathbf{Z}_{i}-\mathbf{W}_{i},\nabla_{\mathbf{X}_{i}^{\ast}}\widetilde{f}_{\mathcal{C}_{i}}(\mathbf{Z}_{i};\mathbf{X})-\nabla_{\mathbf{X}_{i}^{\ast}}\widetilde{f}_{\mathcal{C}_{i}}(\mathbf{W}_{i};\mathbf{X})\right\rangle \\
\qquad\qquad\geq c_{\tau_{i}}\left(\mathbf{X}\right)\,\left\Vert \mathbf{Z}_{i}-\mathbf{W}_{i}\right\Vert _{F}^{2},\quad\forall\mathbf{Z}_{i},\mathbf{W}_{i}\in\mathcal{X}_{i}.
\end{array}
\]
}\textcolor{blue}{{} }\vspace{-0.3cm}

\section{Extensions and related works \label{sec:Extensions-and-Generalizations}}

The key idea in the proposed SCA schemes, e.g., (\ref{eq:best-reponse_c_i_convex}),
is to convexify the nonconvex part of $U$ via partial linearization
of $\sum_{j\in\mathcal{C}_{-i}}f_{j}(\mathbf{X})$, resulting in the
term $\left\langle \mathbf{\Pi}_{\mathcal{C}_{i}}(\mathbf{X}^{n}),\,\Delta\mathbf{X}_{i}\right\rangle $.
In the same spirit of \cite{MarksWright78,Chiang-WeiTan-PalomarOneil-Julian_ITWC-GP,RazaviyaynHongLuo_subOct12arxiv},
it is not difficult to show that one can generalize this idea and
replace the linear term $\left\langle \mathbf{\Pi}_{\mathcal{C}_{i}}(\mathbf{X}^{n}),\,\Delta\mathbf{X}_{i}\right\rangle $
in (\ref{eq:convex_approx_of_fi_on_Ci_convex}) with a nonlinear scalar
function $\Pi_{\mathcal{C}_{i}}(\bullet;\mathbf{X}^{n}):\mathcal{X}_{i}\ni\mathbf{X}_{i}\mapsto\Pi_{\mathcal{C}_{i}}(\mathbf{X}_{i};\mathbf{X}^{n})$.
All the results presented so far are still valid provided that $\Pi_{\mathcal{C}_{i}}(\bullet;\mathbf{X}^{n})$
enjoys the following properties: for all $\mathbf{X}^{n}\in\mathcal{X}$,\smallskip

\noindent P1) $\Pi_{\mathcal{C}_{i}}(\bullet;\mathbf{X}^{n})$ is
$\mathbb{R}$-continuously differentiable on $\mathcal{X}_{i}$; 

\noindent P2) $\nabla_{\mathbf{X}_{i}^{\ast}}\Pi_{\mathcal{C}_{i}}(\mathbf{X}_{i}^{n};\mathbf{X}^{n})=\sum_{j\in\mathcal{C}_{-i}}\nabla_{\mathbf{X}_{i}^{\ast}}f_{j}(\mathbf{X}^{n})$;

\noindent P3) $\nabla_{\mathbf{X}_{i}^{\ast}}\Pi_{\mathcal{C}_{i}}(\mathbf{X}_{i}^{n};\bullet)$
is uniformly Lipschitz on $\mathcal{X}$;

\noindent P4) $\Pi_{\mathcal{C}_{i}}(\mathbf{X}_{i};\mathbf{X}^{n})$
is continuous in $(\mathbf{X}_{i};\mathbf{X}^{n})\in\mathcal{X}_{i}\times\mathcal{X}.$\smallskip

Similar conditions can be written in the real case for the nonlinear
function ${\pi}_{\mathcal{C}_{i}}(\bullet;\mathbf{x}^{n}):\mathcal{K}_{i}\ni\mathbf{x}_{i}\mapsto\pi_{\mathcal{C}_{i}}(\mathbf{x}_{i};\mathbf{x}^{n})$
replacing the linear pricing $\boldsymbol{{\pi}}_{\mathcal{C}_{i}}^{T}\mathbf{x}_{i}$.
It is interesting to compare P1-P3 with conditions in \cite{MarksWright78,Chiang-WeiTan-PalomarOneil-Julian_ITWC-GP,RazaviyaynHongLuo_subOct12arxiv}.
First of all, our conditions do not require that the approximation
function is a global upper bound of the original sum-utility function,
a constraint that remains elusive for sum-utility problems with no
special structure. Second, even when the aforementioned constraint
can be met, it is not always guaranteed that the resulting convex
subproblems are decomposable across the users, implying that a centralized
implementation might be required. Third, SCA algorithms \cite{MarksWright78,Chiang-WeiTan-PalomarOneil-Julian_ITWC-GP,RazaviyaynHongLuo_subOct12arxiv},
even when distributed, are generally \emph{sequential} schemes (unless
the sum-utility has a special structure). On the contrary, the algorithms
proposed in this paper do not suffer from any of the above drawbacks,
which enlarges substantially the class of (large scale) nonconvex
problems solvable using our framework.\vspace{-0.1cm}

\section{Applications and Numerical Results\label{sec:Applications}}

In this section, we customize the proposed decomposition framework
to the SISO and MIMO sum-rate maximization problems introduced in
(\ref{eq:SISO_formulation}) and (\ref{eq:MIMO_formulation_rate}),
respectively, and compare the resulting new algorithms with state-of-the-art
schemes \cite{Papandriopoulos_EvansIT09,SchmidtShiBerryHonigUtschick-SPMag,KimGiannakisIT11,ShiRazaviyaynLuoHe-TSP11,RazaviyaynHongLuo_subOct12arxiv}.
Quite interestingly, our algorithms are shown to outperform current
schemes, in terms of convergence speed and computational effort, while
reaching the same sum-rate. It is worth mentioning that this was not
obvious at all, because algorithms in \cite{Papandriopoulos_EvansIT09,SchmidtShiBerryHonigUtschick-SPMag,KimGiannakisIT11,ShiRazaviyaynLuoHe-TSP11,RazaviyaynHongLuo_subOct12arxiv}
are ad-hoc schemes for the sum-rate problem, whereas our framework
has been introduced for general sum-utility problems.  \vspace{-0.2cm}

\subsection{Sum-Rate Maximization over SISO ICs\label{sub:Sum-Rate-Maximization-SISO}}

Consider the social problem (\ref{eq:SISO_formulation}), with $f_{i}(x)=w_{i}\, x$,
where $w_{i}$ are positive given weights; to avoid redundant constraints,
let also assume w.l.o.g. that all the columns of $\mathbf{W}_{i}$
are linearly independent. We describe next two alternative decompositions
for (\ref{eq:SISO_formulation}) corresponding to differ choices of
$\mathcal{I}_{f}$ and $\{\mathcal{C}_{i}\}$.

\subsubsection{Decomposition \#1$-$Pricing Algorithms}

Since each user's rate $r_{i}(\mathbf{p}_{i},\mathbf{p}_{-i})$ is
concave in $\mathbf{p}_{i}\in\mathcal{P}_{i}$, a natural choice is
$\mathcal{I}_{f}=\mathcal{I}$ and $\mathcal{C}_{i}=\{i\}$, which
leads to the following class of strongly concave subproblems {[}cf.
(\ref{eq:convex_approx_of_fi_on_Ci}){]}: given $\mathbf{p}^{n}=(\mathbf{p}_{i}^{n})_{i=1}^{I}$
and choosing $\mathbf{H}_{i}(\mathbf{p}^{n})=\mathbf{I}$, the best-response
of user $i$ is \vspace{-0.1cm} 
\[
\hspace{-0.2cm}\begin{array}{l}
\hat{\mathbf{p}}{}_{i}(\mathbf{p}^{n})\triangleq\\
\underset{\begin{array}{l}
\mathbf{p}_{i}\in\mathcal{P}_{i}\end{array}}{\mbox{argmax}}\!\!\!\!\left\{ w_{i}\, r_{i}(\mathbf{p}_{i},\mathbf{p}_{-i}^{n})-\boldsymbol{\pi}_{i}(\mathbf{p}^{n})^{T}\mathbf{p}_{i}-\frac{{\tau_{i}}}{2}\left\Vert \mathbf{p}_{i}-\mathbf{p}_{i}^{n}\right\Vert ^{2}\right\} 
\end{array}\vspace{-0.1cm},
\]
where $\boldsymbol{\pi}_{i}(\mathbf{p}^{n})\triangleq(\pi_{ik}(\mathbf{p}^{n}))_{k=1}^{N}$
is the pricing factor, given by 
\begin{equation}
\hspace{-0.1cm}\pi_{i,k}(\mathbf{p}^{n})\triangleq-\sum_{j\in\mathcal{N}_{i}}w_{j}\,|H_{ji}\left(k\right)|^{2}\,\frac{\texttt{{snr}}_{jk}^{n}}{(1+\texttt{{snr}}_{jk}^{n})\cdot\texttt{{mui}}_{jk}^{n}};\label{eq:price_k}
\end{equation}
$\mathcal{N}_{i}$ denotes the set of neighbors of user $i$, i.e.,
the set of users $j$'s which user $i$ interferers with; and $\texttt{{snr}}_{jk}^{n}$
and $\texttt{{mui}}_{jk}^{n}$ are the SINR and the multiuser interference-plus-noise
power experienced by user $j$, generated by the power profile $\mathbf{p}^{n}$:
\[
\texttt{{snr}}_{jk}^{n}\triangleq\dfrac{{|H_{jj}\left(k\right)|^{2}}p_{jk}^{n}}{\texttt{{mui}}_{jk}^{n}},\,\,\texttt{{mui}}_{jk}^{n}\triangleq\sigma_{jk}^{2}+\sum_{i\neq j}{|H_{ji}\left(k\right)|^{2}}p_{ik}^{n}.\vspace{-0.2cm}
\]
 The best-response $\hat{\mathbf{p}}{}_{i}(\mathbf{p}^{n})$ can be
computed in closed form (up to the multipliers associated with the
inequality constraints in $\mathcal{P}_{i}$) according to the following
multi-level waterfilling-like expression \cite{Scutari-Facchinei-Pang-Palomar_IT_PI}:
\begin{equation}
\!\!\begin{array}{l}
\hat{\mathbf{p}}{}_{i}(\mathbf{p}^{n})\triangleq\left[\dfrac{{1}}{2}\,\mathbf{p}_{i}^{n}\circ\left(\mathbf{1}-(\texttt{\textbf{{snr}}}_{i}^{n})^{-1}\right)+\right.\\
\left.\!\!-\dfrac{{1}}{2\,\tau_{i}}\left(\boldsymbol{{\tilde{\mu}}}_{i}\!-\!\sqrt{\left[\boldsymbol{{\tilde{\mu}}}_{i}-\tau_{i}\,\mathbf{p}_{i}^{n}\circ\left(\mathbf{1}+(\texttt{\textbf{{snr}}}_{i}^{n})^{-1}\right)\right]^{2}+4\tau_{i}w_{i}\mathbf{1}}\right)\right]^{+}
\end{array}\label{eq:WF_power}
\end{equation}
where $\circ$ denotes the Hadamard product, $(\texttt{\textbf{{snr}}}_{i}^{n})^{-1}\triangleq(1/\texttt{{snr}}_{ik}^{n})_{k=1}^{N}$
and $\boldsymbol{{\tilde{\mu}}}_{i}\triangleq\boldsymbol{\pi}_{i}(\mathbf{p}^{n})+\mathbf{W}_{i}^{T}\boldsymbol{{\mu}}_{i}$,
with the multiplier vector $\boldsymbol{{\mu}}_{i}$ chosen to satisfy
the nonlinear complementarity condition (CC) $\mathbf{0}\leq\boldsymbol{{\mu}}_{i}\perp\mathbf{I}_{i}^{\max}-\mathbf{W}_{i}\hat{\mathbf{p}}{}_{i}(\mathbf{p}^{n})\geq\mathbf{0}$.
The optimal $\boldsymbol{{\mu}}_{i}$ satisfying the CC can be efficiently
computed (in a finite number of steps) using a multiple nested bisection
method as described in \cite[Alg. 6]{Scutari-Facchinei-Pang-Palomar_IT_PI};
we omit the details because of the space limitation. Note that, in
the presence of the power budget constraint only (as in \cite{Papandriopoulos_EvansIT09,SchmidtShiBerryHonigUtschick-SPMag,ShiRazaviyaynLuoHe-TSP11}),
$\boldsymbol{{\mu}}_{i}$ reduces to a scalar quantity $\mu_{i}$
such that $0\leq{\mu}_{i}\perp P_{i}-\mathbf{1}^{T}\hat{\mathbf{p}}{}_{i}(\mathbf{p}^{n})\geq0$,
whose solution can be obtained using the classical bisection algorithms
(or the methods in \cite{Palomar-Fonollosa_SP05-WF-algo}). 

Given $\hat{\mathbf{p}}{}_{i}(\mathbf{p}^{n})$, one can now use any
of the algorithms introduced in Sec. \ref{sub:Jacobi-Distributed-Pricing}.
For instance, a good candidate is the exact Jacobi scheme with diminishing
step-size (Algorithm \ref{alg:PJA}), whose convergence is guaranteed
if, e.g., rules in (\ref{eq:step-size_1}) or (\ref{eq:step-size_2})
are used for the sequence $\{\gamma^{n}\}$ (Theorem \ref{Theorem_convergence_Jacobi}).
Note that the proposed algorithm is fairly distributed. Indeed, given
the interference generated by the other users {[}and thus the MUI
coefficients $\texttt{{mui}}_{jk}^{n}${]} and the current interference
price $\boldsymbol{\pi}_{i}(\mathbf{p}^{n})$, each user can efficiently
and locally compute the optimal power allocation $\hat{\mathbf{p}}{}_{i}(\mathbf{p}^{n})$
via the waterfilling-like expression (\ref{eq:WF_power}). The estimation
of the prices $\pi_{ik}(\mathbf{p}^{n})$ requires however some signaling
among nearby users. Interestingly, the pricing expression in (\ref{eq:price_k})
as well as the signaling overhead necessary to compute it coincides
with that in \cite{SchmidtShiBerryHonigUtschick-SPMag}. But, because
of their sequential nature, algorithms in \cite{SchmidtShiBerryHonigUtschick-SPMag}
require more CSI exchange in the network then our \emph{simultaneous}
schemes.

\subsubsection{Decomposition \#2$-$DC Algorithms}

An alternative class of algorithms for the sum-rate maximization problem
under consideration can be obtained exploring the D.C. nature of the
rate functions (cf. Example $\#$4 in Sec. \ref{sub:Jacobi-based-schemes}).
The sum-rate can indeed be decomposed as the sum of a concave and
convex function, namely $U(\mathbf{p})=f_{1}(\mathbf{p})+f_{2}(\mathbf{p})$,
where 
\begin{eqnarray*}
f_{1}(\mathbf{p}) & \triangleq & \sum_{i}w_{i}{\displaystyle {\sum_{k}}}\log(\sigma_{i,k}^{2}+\sum_{j}\left|H_{ij}\left(k\right)\right|^{2}p_{jk})\\
f_{2}(\mathbf{p}) & \triangleq & -\sum_{i}w_{i}\sum_{k}\log(\sigma_{i,k}^{2}+\sum_{j\neq i}\left|H_{ij}\left(k\right)\right|^{2}p_{jk}),
\end{eqnarray*}
which is an instance of (\ref{eq:DC_programming}) with $\mathcal{I}_{f}=\{1,2\}$.
A natural choice of $\mathcal{C}_{i}$ is then $\mathcal{C}_{i}=\{1\}$
for all $i\in\mathcal{I}$, resulting in the best-response:\vspace{-0.1cm}{\small{
\[
\!\!\widetilde{{\mathbf{p}}}_{i}(\mathbf{p}^{n})\!\triangleq\!\!\!\underset{\begin{array}{l}
\mathbf{p}_{i}\in\mathcal{P}_{i}\end{array}}{\mbox{argmax}}\!\!\!\!\left\{ f_{1}(\mathbf{p}_{i},\mathbf{p}_{-i}^{n})-{\boldsymbol{\pi}}_{i}(\mathbf{p}^{n})^{T}\mathbf{p}_{i}-\frac{{\tau_{i}}}{2}\left\Vert \mathbf{p}_{i}-\mathbf{p}_{i}^{n}\right\Vert ^{2}\right\} ,
\]
}}where ${\boldsymbol{\pi}}_{i}(\mathbf{p}^{n})\triangleq({\pi}_{ik}(\mathbf{p}^{n}))_{k=1}^{N}$,
with\vspace{-0.1cm} 
\begin{equation}
\pi_{i,k}(\mathbf{p}^{n})\triangleq-\sum_{j\in\mathcal{N}_{i}}w_{j}\,|H_{ji}\left(k\right)|^{2}\,\frac{1}{\texttt{{mui}}_{jk}^{n}}.\vspace{-0.2cm}\label{eq:pricing_DC}
\end{equation}
We remark that the best-response $\widetilde{{\mathbf{p}}}_{i}(\mathbf{p}^{n})$
can be efficiently computed by a fixed-point iterate, in the same
spirit of \cite{Papandriopoulos_EvansIT09}; we omit the details because
of the space limitation. Note that the communication overhead to compute
the prices (\ref{eq:price_k}) and (\ref{eq:pricing_DC}) is the same,
but the computation of $\widetilde{{\mathbf{p}}}_{i}(\mathbf{p}^{n})$
requires more CSI exchange in the network than that of $\hat{\mathbf{p}}{}_{i}(\mathbf{p}^{n})$,
since each user $i$ also needs to estimate the cross-channels $\{\left|H_{ji}\left(k\right)\right|^{2}\}_{j\in\mathcal{N}_{i}}$. 

\noindent \textbf{Numerical Example}\emph{. }We compare now Algorithm
1 based on the best-response $\hat{\mathbf{p}}{}_{i}(\mathbf{p}^{n})$
in (\ref{eq:WF_power}) (termed SJBR), with those proposed in \cite{Papandriopoulos_EvansIT09}
{[}termed SCALE and SCALE one-step, the latter being a simplified
version of SCALE where instead of solving the fixed-point equation
(16) in \cite{Papandriopoulos_EvansIT09}, only one iteration of (16)
is performed{]}, \cite{SchmidtShiBerryHonigUtschick-SPMag} (termed
MDP), \cite{ShiRazaviyaynLuoHe-TSP11} (termed WMMSE). Since in the
aforementioned papers only power budget constraints can be dealt with,
to allow the comparison, we simplified the sum-rate maximization problem
described above and considered only power budget constraints (and
all $w_{i}=1$). We assume the same power budget $P_{i}=P$, noise
variances $\sigma_{ik}^{2}=\sigma^{2}$, and $\texttt{{snr}}=P/\sigma^{2}=3$dB
for all the users. We simulated SISO frequency channels with $N=64$
subcarriers; the channels are generated as FIR filters of order $L=10$,
whose taps are i.i.d. Gaussian random variables with zero mean and
variance $1/(d_{ij}^{3}(L+1)^{2})$, where $d_{ij}$ is the distance
between the transmitter $j$ and the receiver $i$. All the algorithms
are initialized by choosing the uniform power allocation, and are
terminated when (the absolute value) of the sum-utility error in two
consecutive rounds becomes smaller than $1e$-$6$. The accuracy in
the bisection loops (required by all methods) is set to $1e$-$7$.
In our algorithm, we used rule (\ref{eq:step-size_1}) with $\epsilon=1e$-$2$
and set all $\tau_{i}=0$. In Fig. \ref{fig1}, we plot the average
number of iterations required by the aforementioned algorithms to
converge versus the number of users; the average is taken over $100$
independent channel realizations; we set $d_{ij}/d_{ii}=3$ and $d_{ij}=d_{ji}$
and $d_{ii}=d_{jj}$ for all $i$ and $j\neq i$. As benchmark, we
also plot two instances of proximal conditional gradient algorithms
\cite{Bertsekas_Book-Parallel-Comp}, which can be interpreted as
special cases of our SJBR with $\mathcal{C}_{i}=\emptyset$ for all
$i\in\mathcal{I}$ (cf. Ex. \#1 in Sec. \ref{sub:Jacobi-based-schemes}).
In one instance {[}termed Gradient (SJBR tuning){]} we set the free
parameters $\tau_{i}$ and $\epsilon$ as in SJBR, whereas in the
other one {[}termed Gradient (opt. tuning){]} we chose $\tau_{i}=50$
for all $i\in\mathcal{I}$ and $\epsilon=1e$-$2$, which leads experimentally
to the fastest behavior of the gradient algorithm. 
\begin{figure}
\vspace{-0.2cm}
\includegraphics[height=6.3cm]{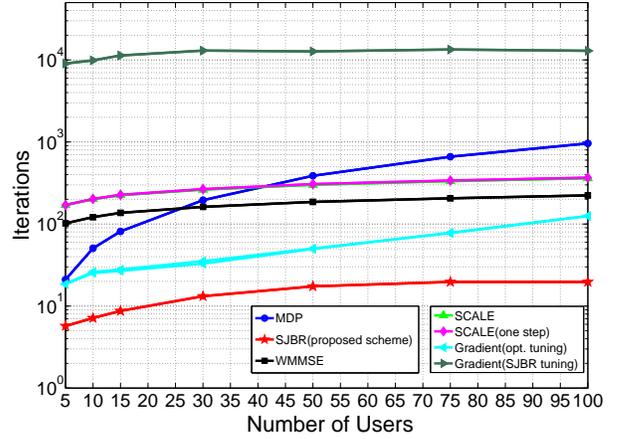}\vspace{-0.1cm}
 \caption{\textcolor{black}{\footnotesize{Average number of iterations versus
number of users in SISO frequency-selective ICs. Note that all algorithms
are simultaneous except MDP; this means that, at each iteration, in
MDP there is only one user updating his strategy, whereas in the other
algorithms all users do so). }}}
\vspace{-0.5cm}

{\small{\label{fig1}}} 
\end{figure}

All the algorithms reach the same average sum-rate (that thus is not
reported here, see \textcolor{black}{\cite{Song_Report13}}), but
their convergence behavior is quite different. The figure clearly
shows that our SJBR outperforms all the others (note that SCALE, WMMSE,
and the proximal gradient are also simultaneous-based schemes). For
instance, the gap with the WMMSE is about one order of magnitude,
for all the network sizes considered in the experiment, while the
gap with MDP is up to three orders of magnitude. The good behavior
of our scheme has been observed also for other choices of $d_{ij}/d_{ii}$,
termination tolerances, and step-size rules; we cannot present here
more experiments because of space limitation; we refer the interested
reader to the technical report\textcolor{black}{{} \cite{Song_Report13}
}for more numerical results. Note that SJBR, SCALE one-step, WMMSE,
MDP, and gradient schemes have similar per-user computational complexity,
whereas SCALE is much more demanding and is not appealing for a real-time
implementation. Therefore, Fig. \ref{fig1} provides also a rough
indication of the per-user cpu time of SJBR, SCALE one-step, WMMSE,
and gradient algorithms.

It is also interesting to compare the proposed algorithm with gradient
schemes. A first natural question is whether the partial linearization
(as performed in SJBR) really improves the convergence speed of the
algorithm. The answer is given by the comparison in Fig. \ref{fig1}
between SJBR and ``Gradient (SJBR tuning)''. One can see that, under
the same choice of $\{\gamma^{n}\}$ and $(\tau_{i})_{i=1}^{I}$,
the former is almost three order of magnitude faster then the latter,
for all the network sizes considered in the experiment. If an independent,
ad hoc tuning of $\{\gamma^{n}\}$ and $(\tau_{i})_{i=1}^{I}$ is
performed for the gradient algorithm, the gap reduces up to one order
of magnitude, still in favor of SJBR. This result supports the intuition
motivating this work: preserving the structure of the problem via
a partial linearization can significantly improve the convergence
speed of the algorithm. 

The comparison with gradient algorithms also reveals a well-known
issue of these schemes: the convergence behavior strongly depends
on the choice of the step-size sequence $\{\gamma^{n}\}$ and the
proximal gains $\tau_{i}$. It is then natural to ask whether also
the proposed algorithms suffer from the same drawback. To answer this
question, in Fig. \ref{fig2} we compare the convergence behavior
of the proximal condition gradient algorithm with that of SJBR, using
the step-size rule (\ref{eq:step-size_1}), but changing the free
parameter $\epsilon\in(0,1)$ by several orders of magnitude. For
gradient schemes, we considered two choices of $\tau_{i}$, namely:
$\tau_{i}=0$ and $\tau_{i}=50$ (as in Fig. \ref{fig1}); the latter
resulting in the experimentally fastest behavior of gradient schemes
(see Fig. \ref{fig1}). More specifically, in Fig. \ref{fig2}, we
plot the average number of iterations needed to reach convergence
within the accuracy of $1e$-$6$ versus $\epsilon\in(0,1)$, for
different number of users (the rest of the setting is as in Fig. \ref{fig1}).
The figure clearly shows that, differently from gradient algorithms,
the convergence behavior of our scheme appears to be almost independent
of the choice of $\epsilon$. This is a very desirable feature that
lets one avoid the expensive and difficult tuning of the step-size,
thus making the proposed algorithms a very good candidate in many
applications. We remark one more time that the gradient method is
very sensitive to the choice of parameters; indeed, based on further
simulations that we do not report here for lack of space, the behavior
of the gradient method is very sensitive to the number of users and
characteristics of the network (SNR, pair distances, etc...) and its
optimal behavior requires different tunings of parameters each time.
\begin{figure}[h]
\vspace{-0.3cm}
\includegraphics[height=6.5cm]{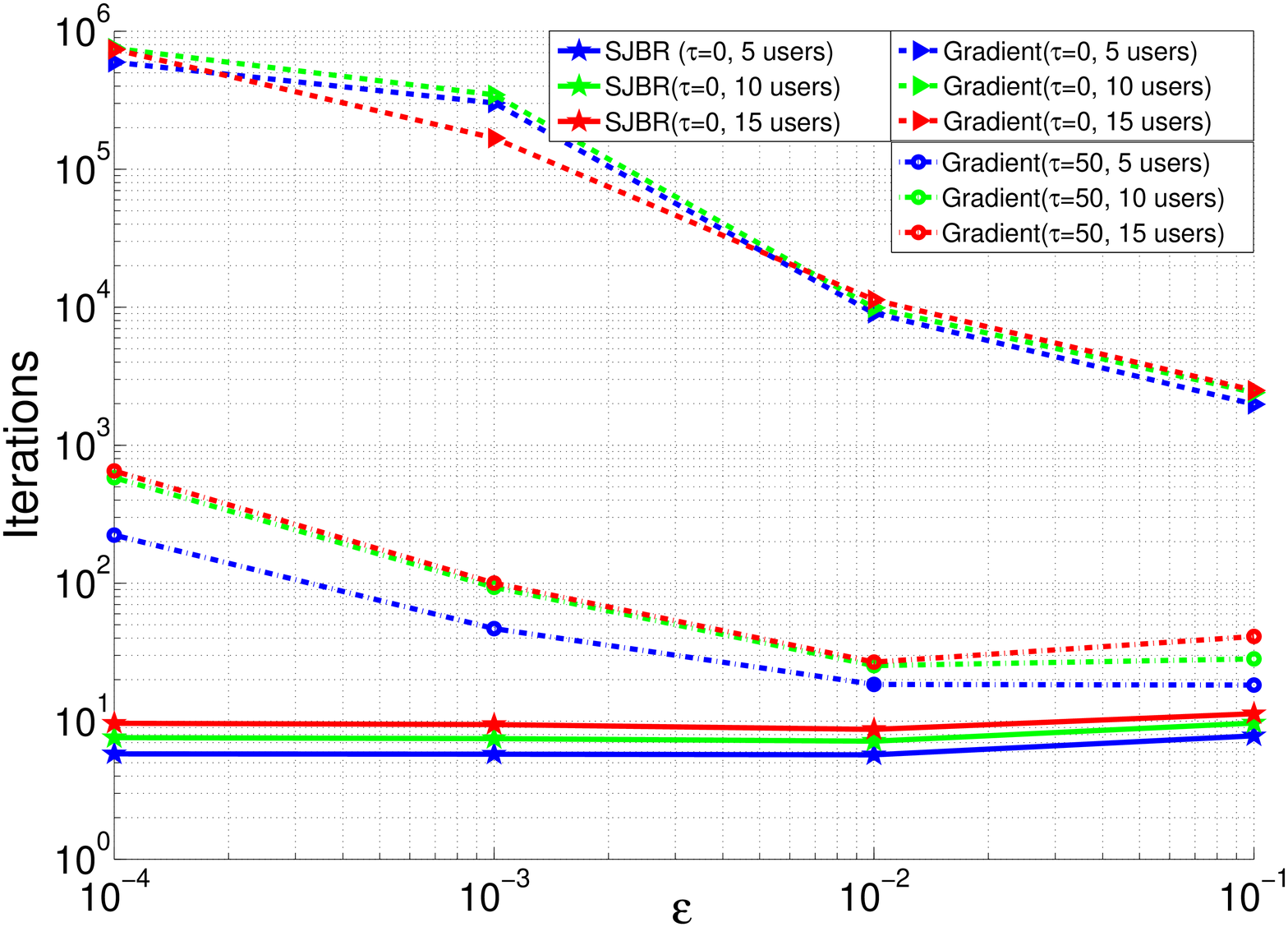}\vspace{-0.2cm}
 \caption{\textcolor{black}{\footnotesize{Proximal conditional gradient algorithms
versus SJBR: Average number of iterations versus $\epsilon\in(0,1)$
{[}cf. (\ref{eq:step-size_1}){]}. }}}
\vspace{-0.4cm}

{\small{\label{fig2}}} 
\end{figure}

\subsection{Sum-Rate Maximization over MIMO ICs}

Let us focus now on the MIMO formulation (\ref{eq:MIMO_formulation_rate}),
assuming $f_{i}(x)=w_{i}\, x$, with $w_{i}>0$.

\subsubsection{Decomposition \#1: Pricing Algorithms}

Choosing $I_{f}=I$, $\mathcal{C}_{i}=\{i\}$, and $\mathcal{H}_{i}(\mathbf{Q}^{n})=\mathbf{I}$,
the best-response of user $i$ is

\begin{equation}
\hspace{-0.2cm}\begin{array}{l}
\hat{\mathbf{Q}}_{i}(\mathbf{Q}^{n},\tau_{i})\triangleq\!\!\!\!\underset{\begin{array}{l}
\mathbf{Q}_{i}\in\mathcal{Q}_{i}\end{array}}{\mbox{argmax}}\!\!\!\!\left\{ w_{i}\, r_{i}(\mathbf{Q}_{i},\mathbf{Q}_{-i}^{n})-\left\langle \mathbf{\Pi}_{i}(\mathbf{Q}^{n}),\mathbf{Q}_{i}-\mathbf{Q}_{i}^{n}\right\rangle \right.\vspace{-0.3cm}\\
\left.\qquad\quad\quad\qquad\qquad\qquad-\tau_{i}\left\Vert \mathbf{Q}_{i}-\mathbf{Q}_{i}^{n}\right\Vert _{F}^{2}\right\} 
\end{array}\label{eq:best-response_MIMO}
\end{equation}
with 
\[
\mathbf{\Pi}_{i}(\mathbf{Q}^{n})\triangleq\sum_{j\in\mathcal{N}_{i}}w_{j}\,\mathbf{H}_{ji}^{H}\widetilde{\mathbf{R}}_{j}(\mathbf{Q}_{-j}^{n})\mathbf{H}_{ji},
\]
where $\mathcal{N}_{i}$ is defined as in the SISO case, and 
\[
\widetilde{\mathbf{R}}_{j}(\mathbf{Q}_{-j}^{n})\triangleq\mathbf{R}_{j}(\mathbf{Q}_{-j}^{n})^{-1}-(\mathbf{R}_{j}(\mathbf{Q}_{-j}^{n})+\mathbf{H}_{jj}\mathbf{Q}_{j}^{n}\mathbf{H}_{jj}^{H})^{-1}.
\]

Note that, once the price matrix $\mathbf{\Pi}_{i}(\mathbf{Q}^{n})$
is given, the best-response $\hat{\mathbf{Q}}_{i}(\mathbf{Q}^{n},\tau_{i})$
can be computed locally by each user solving a convex optimization
problem. Moreover, for some specific structures of the feasible sets
$\mathcal{Q}_{i}$, the case of full-column rank channel matrices
$\mathbf{H}_{i}$, and $\tau_{i}=0$, a solution in closed form (up
to the multipliers associated with the power budget constraints) is
also available \cite{KimGiannakisIT11}. Given $\hat{\mathbf{Q}}_{i}(\mathbf{Q}^{n},\tau_{i})$,
one can now use any of the algorithms introduced in Sec. \ref{sec:The-Complex-Case}.
\textcolor{black}{To the best of our knowledge, our schemes are the
first class of }\textcolor{black}{\emph{best-response Jacobi}}\textcolor{black}{{}
(inexact) algorithms for MIMO IC systems based on }\textcolor{black}{\emph{pricing}}\textcolor{black}{{}
with provable convergence.\smallskip}

\noindent \emph{Complexity Analysis and Message Exchange}. It is
interesting to compare the computational complexity and signaling
(i.e., message exchange) of our algorithms, e.g., Algorithm 1 based
on the best-response $\hat{\mathbf{Q}}_{i}(\mathbf{Q}^{n},\tau_{i})$
(termed MIMO-SJBR) with those of the schemes proposed in the literature
for a similar problem, namely the MIMO-MDP \cite{SchmidtShiBerryHonigUtschick-SPMag,KimGiannakisIT11},
and the MIMO-WMMSE \cite{ShiRazaviyaynLuoHe-TSP11}. We assume that
all channel matrices $\mathbf{H}_{ii}$'s are full-column rank, and
set $\tau_{i}=0$ in (\ref{eq:best-response_MIMO}). For the purpose
of complexity analysis, since all algorithms include a similar bisection
step which generally takes few iterations, we will ignore this step
in the computation of the complexity (as in \cite{ShiRazaviyaynLuoHe-TSP11}).
Also, WMMSE and SJBR are simultaneous schemes, while MDP is sequential;
we then compare the algorithms by given the \emph{per-round} \emph{complexity},
where one round means one update of all users. Denoting by $n_{T}$
(resp. $n_{R}$) the number of antennas at each transmitter (resp.
receiver), the\emph{ }computational complexity of the algorithms is:\smallskip

\noindent $\bullet$ MIMO-MDP: $\mathcal{O}\!\left(I^{2}(n_{T}n_{R}^{2}+n_{T}^{2}n_{R}+n_{R}^{3})+I\, n_{T}^{3}\right)$ 

\noindent $\bullet$ MIMO-WMMSE: $\!\mathcal{O}\!\left(I^{2}(n_{T}n_{R}^{2}+n_{T}^{2}n_{R}+n_{T}^{3})+I\, n_{R}^{3}\right)\!$
\cite{ShiRazaviyaynLuoHe-TSP11} 

\noindent $\bullet$ MIMO-SJBR: $\mathcal{O}\!\left(I^{2}(n_{T}n_{R}^{2}+n_{T}^{2}n_{R})+I(n_{T}^{3}+n_{R}^{3})\right)$.\smallskip 

It is clear that the complexity of the three algorithms is very similar,
and same in order in the case in which $n_{T}=n_{R}$($\triangleq n$),
given by $\mathcal{O}(I^{2}n^{3})$.

We remark that, in a real system, the MUI covariance matrices $\mathbf{R}_{i}(\mathbf{Q}_{-i})$
come from an estimation process. It is thus interesting to understand
how the complexity changes when the computation of $\mathbf{R}_{i}(\mathbf{Q}_{-i})$
from $\mathbf{R}_{n_{i}}+\sum_{j\neq i}\mathbf{H}_{ij}\mathbf{Q}_{j}\mathbf{H}_{ij}^{H}$
is not included in the analysis. We obtain the following:\smallskip

\noindent $\bullet$ MIMO-MDP: $\mathcal{O}\!\left(I^{2}(n_{T}n_{R}^{2}+n_{T}^{2}n_{R}+n_{R}^{3})+I\, n_{T}^{3}\right)$

\noindent $\bullet$ MIMO-WMMSE: $\mathcal{O}\!\left(I^{2}(n_{T}^{2}n_{R}+n_{T}^{3})+I(n_{R}^{3}+n_{T}n_{R}^{2})\right)$

\noindent $\bullet$ MIMO-SJBR: $\mathcal{O}\!\left(I^{2}(n_{T}n_{R}^{2}+n_{T}^{2}n_{R})+I(n_{T}^{3}+n_{R}^{3})\right)$.
\smallskip

Finally, if one is interested in the time necessary to complete one
iteration, it can be shown that it is proportional to the above complexity
divided by $I$.

As far as the communication overhead is concerned, the same remarks
we made about the schemes described in the SISO setting, apply also
here for the MIMO case. The only difference is that now the users
need to exchange a (pricing) matrix rather than a vector, resulting
in $\mathcal{O}(I^{2}\, n_{R}^{2})$ amount of message exchange per-iteration
for all the algorithms. 

\subsubsection{Decomposition \#2$-$WMMSE Algorithms}

In \cite{ShiRazaviyaynLuoHe-TSP11}, the authors showed that the MIMO
problem (\ref{eq:MIMO_formulation_rate}) (under power constraints
only) is equivalent to the following sum-MSE minimization: writing
$\mathbf{Q}_{i}=\mathbf{V}_{i}\mathbf{V}_{i}^{H}$, $\mathbf{V}\triangleq(\mathbf{V}_{i})_{i=1}^{I}$,
and introducing the auxiliary matrix variables $\mathbf{U}\triangleq(\mathbf{U}_{i})_{i=1}^{I}$,
$\mathbf{W}\triangleq(\mathbf{W})_{i=1}^{I}$, 
\begin{equation}
\!\!\!\!\!\begin{array}{ll}
\underset{\mathbf{W},\mathbf{U},\mathbf{V}}{\mbox{min}} & \!\!\!\!\! f\left(\mathbf{W},\mathbf{U},\mathbf{V}\right)\triangleq{\displaystyle {\sum_{i\in\mathcal{I}}}}w_{i}\left(\text{{tr}}\left(\mathbf{W}_{i}\,\mathbf{E}_{i}(\mathbf{U},\mathbf{V})\right)\!-\!\log\det(\mathbf{W}_{i})\right)\\[5pt]
\,\,\,\,\,\text{s.t.} & \!\!\!\!\!\text{{tr}}(\mathbf{V}_{i}\mathbf{V}_{i}^{H})\leq P_{i},\,\,\mathbf{W}_{i}\succeq\mathbf{0},\quad\forall i\in\mathcal{I},
\end{array}\label{eq:MSE_formulation}
\end{equation}
where $\mathbf{E}_{i}(\mathbf{U},\mathbf{V})$ is the MSE matrix at
the receiver $i$ (see (3) in \cite{ShiRazaviyaynLuoHe-TSP11}). The
formulation (\ref{eq:MSE_formulation}) has some desirable properties,
namely: i) $f\left(\mathbf{W},\mathbf{U},\mathbf{V}\right)$ is continuously
($\mathbb{R}$-)differentiable with Lipschitz continuous (conjugate)
gradient on the feasible set; ii) $f\left(\mathbf{W},\mathbf{U},\mathbf{V}\right)$
is convex in each variables $\mathbf{W}$, $\mathbf{U}$, $\mathbf{V}$;
iii) the minimization of $f\left(\mathbf{W},\mathbf{U},\mathbf{V}\right)$
w.r.t. to each $\mathbf{W}$, $\mathbf{U}$, $\mathbf{V}$ can be
performed in parallel by the users; and iv) the optimal solutions
of the individual minimizations are available in closed form, see
\cite{ShiRazaviyaynLuoHe-TSP11} for details. We will denote such
optimal solutions as $\hat{{\mathbf{W}}}_{i}(\mathbf{U},\mathbf{V})$,
$\hat{{\mathbf{U}}}_{i}(\mathbf{U},\mathbf{V})$, and $\hat{{\mathbf{V}}}_{i}(\mathbf{U},\mathbf{W})$,
for all $i\in\mathcal{I}$, where we made explicit the dependence
on the variables that are kept fixed. In \cite{ShiRazaviyaynLuoHe-TSP11}
the authors proposed to use the (Gauss-Seidel) block coordinate descent
method to solve (\ref{eq:MSE_formulation}), resulting in the so-called
MIMO-WMMSE algorithm. 

It is not difficult to see that the formulation (\ref{eq:MSE_formulation})
can be cast into our framework, resulting in the following best-response
mapping for each user $i$: $\hat{{\mathbf{X}}}_{i}\left(\mathbf{W}^{n},\mathbf{U}^{n},\mathbf{V}^{n}\right)\triangleq\left(\hat{{\mathbf{W}}}_{i}(\mathbf{U}^{n},\mathbf{V}^{n}),\hat{{\mathbf{U}}}_{i}(\mathbf{U}^{n},\mathbf{V}^{n}),\hat{{\mathbf{V}}}_{i}(\mathbf{U}^{n},\mathbf{W}^{n})\right)$.
We can then compute a stationary solution of (\ref{eq:MSE_formulation})
and thus (\ref{eq:MIMO_formulation_rate}) using any of the \emph{Jacobi}
algorithms introduced in the previous sections based on $\hat{{\mathbf{X}}}_{i}\left(\mathbf{W}^{n},\mathbf{U}^{n},\mathbf{V}^{n}\right)$
(or its inexact computation). Note that the computational complexity
as well as the communication overhead of such algorithms are roughly
the same of those of the MIMO-WMMSE  \cite{ShiRazaviyaynLuoHe-TSP11}.\smallskip

\noindent \textbf{Numerical Example}\emph{. }In Tables I and II we
compare the MIMO-SJBR, the MIMO-MDP \cite{SchmidtShiBerryHonigUtschick-SPMag,KimGiannakisIT11},
and the MIMO-WMMSE \cite{SchmidtShiBerryHonigUtschick-SPMag,KimGiannakisIT11},
in terms of average number of iterations required to reach convergence,
for different number of users, normalized distances $d\triangleq d_{ij}/d_{ii}$
(with $d_{ij}=d_{ji}$ and $d_{ii}=d_{jj}$ for all $i$ and $j\neq i$),
and termination accuracy (namely: $1e$-$3$ and $1e$-$6$). \textcolor{black}{We
considered the following setup. }All the transmitters/receivers are
equipped with $4$ antennas; we simulated uncorrelated fading channel
model, where the coefficients are Gaussian distributed with zero mean
and variance $1/d_{ij}^{3}$; and we set $\mathbf{R}_{n_{i}}=\sigma^{2}\mathbf{I}$
for all $i$, and $\texttt{{snr}}\triangleq P/\sigma^{2}=3$dB. \textcolor{black}{We
used the step-size rule (\ref{eq:step-size_1}) with $\epsilon=1e$-$5$
and $\tau_{i}=0$. }We computed the best-response (\ref{eq:best-response_MIMO})
using the closed form solution \cite{KimGiannakisIT11}. 

In our simulations all the algorithms reached the same average sum-rate.
Given the results in Tables I and II, the following comments are in
order. The proposed SJBR outperforms all the others schemes in terms
of iterations, while having similar (or even better) computational
complexity. Interestingly, the iteration gap with the other schemes
reduces with the distance and the termination accuracy. More specifically:
i) SJBR seems to be much faster than all the other schemes (about
one order of magnitude) when $d_{ij}/d_{ii}=3$ {[}say low interference
scenarios{]}, and just a bit faster (or comparable to MIMO-WMMSE)
when $d_{ij}/d_{ii}=1$ {[}say high interference scenarios{]}; and
ii) SJBR is much faster than all the others, if an high termination
accuracy is set (see Table I). Also, the convergence speed of SJBR
is not affected too much by the number of users. Finally, in our experiments,
we also observed that the performance of SJBR are not affected too
much by the choice of the parameter $\epsilon$ in the (\ref{eq:step-size_1}):
a change of $\epsilon$ of many orders of magnitude leads to a difference
in the average number of iterations which is within 5\%; we refer
the reader to\textcolor{black}{{} \cite{Song_Report13} }for details,
where one can also find a comparison of several other step-size rules.
We must stress however that MIMO-MDP and MIMO-WMMSE do not need any
tuning, which is an advantage with respect to our method. {\setlength{\tabcolsep}{4.0pt} \begin{table}[h]\scriptsize  \begin{tabular}{l|lll|lll|lll}  & \multicolumn{3}{|c|}{$\#\mbox{ of users}=10$} &  \multicolumn{3}{|c}{$\#\mbox{ of users}=50$}& \multicolumn{3}{|c}{$\#\mbox{ of users}=100$}\\   & d=1 & d=2 & d=3 & d=1 & d=2 & d=3 & d=1 & d=2 & d=3 \\  \hline  MDP & 1370.5 & 187 & 54.4 &  4148.5 & 1148 & 348 & 8818 & 1904 & 704\\  WMMSE &169.2 & 68.8 & 53.3 & 138.5 & 115.2 & 76.7& 154.3 & 126.9 & 103.2\\  JSBR & 169.2 & 24.3 & 6.9 & 115.2 & 34.3 & 9.3 & 114.3 & 28.4 & 9.7  \end{tabular} \caption{Average number of iterations (termination accuracy=$1e$-$6$)}  \end{table}}\vspace{-0.5cm}
   {\setlength{\tabcolsep}{4.0pt} \begin{table}[h]\scriptsize  \begin{tabular}{l|lll|lll|lll}  & \multicolumn{3}{|c|}{$\#\mbox{ of users}=10$} &  \multicolumn{3}{|c}{$\#\mbox{ of users}=50$}& \multicolumn{3}{|c}{$\#\mbox{ of users}=100$}\\   & d=1 & d=2 & d=3 & d=1 & d=2 & d=3 & d=1 & d=2 & d=3 \\  \hline  MDP &429.4& 74.3 & 32.8 &1739.5&  465.5 & 202 &3733& 882 & 442.6 \\  WMMSE & 51.6 & 19.2 & 14.7 & 59.6 & 24.9 & 16.3& 69.8 & 26.0 & 19.2\\  JSBR & 48.6 & 9.4 & 4.0 & 46.9 & 12.6 & 5.1 & 49.7 & 12 & 5.5  \end{tabular} \caption{Average number of iterations (termination accuracy=$1e$-$3$)}  \end{table}}\vspace{-0.5cm}

\section{Conclusion\label{sec:Conclusion}\vspace{-0.1cm}}

In this paper, we proposed a novel decomposition framework, based
on SCA, to compute stationary solutions of general nonconvex sum-utility
problems (including social functions of complex variables). The main
result is a new class of convergent \emph{distributed} \emph{Jacobi}
(inexact) best-response algorithms, where all users \emph{simultaneously}
solve (inexactly) a suitably convexified version of the original social
problem. Our framework contains as special cases many decomposition
methods already proposed in the literature, such as gradient algorithms,
and many block-coordinate descent schemes for convex functions. Finally,
we tested our methodology on some sum-rate maximization problems over
SISO/MIMO ICs; our experiments show that our algorithms are faster
than ad-hoc state-of-the-art methods while having the same (user)
computational complexity in the SISO case and similar (or better)
complexity in the MIMO case. Some interesting future directions of
this work are under investigation, e.g., how to adaptively choose
the step-size rule (so that no a-priori tuning is needed), and how
to generalize our framework to scenarios when only long-term channel
statistics are available.\vspace{-0.2cm}

\section*{\textcolor{black}{Acknowledgments}}

\textcolor{black}{The authors wish to thank the Associate Editor,
Prof. Anthony So, and the anonymous reviewers for their valuable comments.
The authors are also deeply grateful to Prof. Tom Luo, Wei-Cheng Liao,
and Yang Yang whose  comments contributed to improve the quality of
the paper.}

\textcolor{black}{The research of Scutari and Song is supported by
the grants NSF No. CNS-1218717 and NSF CAREER No. ECCS-1254739. The
research of Palomar is supported by the Hong Kong RGC 617810 research
grant. The research of Pang is supported by NSF grant No. CMMI 0969600
(awarded to the University of Illinois at Urbana-Champaign).} \vspace{-0.2cm}

\section*{Appendix}

For notational simplicity, in the following we will omit in each $\widehat{\mathbf{x}}_{\mathcal{C}_{i}}(\mathbf{y},\tau_{i})$
{[}and $\widehat{\mathbf{x}}_{\mathcal{C}}(\mathbf{y},\boldsymbol{{\tau}})${]}
the dependence on $\mathcal{C}_{i}$ and ${\tau}_{i}$, and write
$\widehat{\mathbf{x}}_{i}(\mathbf{y})$ {[}and $\widehat{\mathbf{x}}(\mathbf{y})${]};
also, we introduce $f_{\mathcal{C}_{i}}(\mathbf{x}_{i},\mathbf{x}_{-i})\triangleq\sum_{j\in\mathcal{C}_{i}}f_{j}(\mathbf{x}_{i},\mathbf{x}_{-i})$
and $f_{\mathcal{C}_{-i}}(\mathbf{x}_{i},\mathbf{x}_{-i})\triangleq\sum_{j\in\mathcal{C}_{-i}}f_{j}(\mathbf{x}_{i},\mathbf{x}_{-i})$.\vspace{-0.2cm}

\subsection{\label{sec:Appendix:-Proposition_x_map}Proof of Proposition \ref{Prop_x_y} }

\textcolor{black}{Before proving the proposition, let us introduce
the following intermediate result whose proof is a consequence of
assumptions A1-A3 and thus is omitted.}

\begin{lemma}\label{Lemma_f_x_y_properties} \textcolor{black}{Let
$\tilde{{f}}(\mathbf{x};\mathbf{y})\triangleq\sum_{i}\tilde{{f}}_{\mathcal{C}_{i}}(\mathbf{x}_{i};\mathbf{y})$,
with $\tilde{{f}}_{\mathcal{C}_{i}}(\mathbf{x}_{i};\mathbf{y})$ defined
in (\ref{eq:convex_approx_of_fi_on_Ci}). Then the following hold:}

\textcolor{blue}{\noindent}\textcolor{black}{(i) $\tilde{{f}}(\mathbf{\bullet};\mathbf{y})$
is uniformly strongly convex on $\mathcal{K}$ with constant $c_{\boldsymbol{{\tau}}}>0$,
i.e., 
\begin{equation}
\begin{array}{l}
\left(\mathbf{x}-\mathbf{w}\right)^{T}\left(\nabla_{\mathbf{x}}\tilde{{f}}\left(\mathbf{x};\mathbf{y}\right)-\nabla_{\mathbf{x}}\tilde{{f}}\left(\mathbf{w};\mathbf{y}\right)\right)\geq c_{{\boldsymbol{\tau}}}\left\Vert \mathbf{x}-\mathbf{w}\right\Vert ^{2}\end{array},\label{eq:strong_cvx_f_tilde}
\end{equation}
for all $\mathbf{x},\mathbf{w}\in\mathcal{K}$ and given $\mathbf{y}\in\mathcal{K}$;}

\textcolor{blue}{\noindent}\textcolor{black}{(ii) $\nabla_{\mathbf{x}}\tilde{{f}}(\mathbf{x};\mathbf{\bullet})$
is uniformly Lipschitz continuous on $\mathcal{K}$, i.e., there exists
a $0<L_{\nabla\tilde{{f}}}<\infty$ independent on $\mathbf{x}$ such
that 
\begin{equation}
\begin{array}{l}
\left\Vert \nabla_{\mathbf{x}}\tilde{{f}}\left(\mathbf{x};\mathbf{y}\right)-\nabla_{\mathbf{x}}\tilde{{f}}\left(\mathbf{x};\mathbf{w}\right)\right\Vert \end{array}\leq\, L_{\nabla\tilde{{f}}}\,\left\Vert \mathbf{y}-\mathbf{w}\right\Vert ,\label{eq:Lip_grad_L_f}
\end{equation}
for all $\mathbf{y},\mathbf{w}\in\mathcal{K}$ and given $\mathbf{x}\in\mathcal{K}$.}\end{lemma}

We prove now the statements of Proposition \ref{Prop_x_y} in the
following order (c)-(a)-(b)-(d).

\noindent (c): Given $\mathbf{y}\in\mathcal{K}$, by definition,
each $\widehat{\mathbf{x}}_{i}(\mathbf{y})$ is the unique solution
of the problem (\ref{eq:decoupled_problem_i}) and thus satisfies
the minimum principle: for all $\mathbf{z}_{i}\in\mathcal{K}_{i}$,\vspace{-0.2cm}
\begin{equation}
\begin{array}[t]{l}
\left(\mathbf{z}_{i}-\widehat{\mathbf{x}}_{i}(\mathbf{y})\right)^{T}\\
\left(\nabla_{\mathbf{x}_{i}}f_{\mathcal{C}_{i}}\!\left(\widehat{\mathbf{x}}_{i}(\mathbf{y}),\,\mathbf{y}_{-i}\right)+\boldsymbol{{\pi}}_{\mathcal{C}_{i}}(\mathbf{y})+\tau_{i}\,\mathbf{H}_{i}(\mathbf{y})\left(\widehat{\mathbf{x}}_{i}(\mathbf{y})-\mathbf{y}_{i}\right)\right)\geq0.
\end{array}\label{eq:VI_i}
\end{equation}
Summing and subtracting $\nabla_{\mathbf{x}_{i}}f_{\mathcal{C}_{i}}\left(\mathbf{y}_{i},\,\mathbf{y}_{-i}\right)$
in (\ref{eq:VI_i}), choosing $\mathbf{z}_{i}=\mathbf{y}_{i}$, and
using $\boldsymbol{{\pi}}_{\mathcal{C}_{i}}(\mathbf{y})\triangleq\nabla_{\mathbf{x}_{i}}f_{\mathcal{C}_{-i}}(\mathbf{y})$,
we get 
\begin{equation}
\begin{array}[t]{l}
\left(\mathbf{y}_{i}-\widehat{\mathbf{x}}_{i}(\mathbf{y})\right)^{T}\left(\nabla_{\mathbf{x}_{i}}f_{\mathcal{C}_{i}}\!\left(\widehat{\mathbf{x}}_{i}(\mathbf{y}),\,\mathbf{y}_{-i}\right)-\nabla_{\mathbf{x}_{i}}f_{\mathcal{C}_{i}}\!\left(\mathbf{y}_{i},\,\mathbf{y}_{-i}\right)\right)\smallskip\\
\quad+\left(\mathbf{y}_{i}-\widehat{\mathbf{x}}_{i}(\mathbf{y})\right)^{T}\nabla_{\mathbf{x}_{i}}U(\mathbf{y})\smallskip\\
\quad-\tau_{i}\,(\widehat{\mathbf{x}}_{i}(\mathbf{y})-\mathbf{y}_{i})^{T}\,\mathbf{H}_{i}(\mathbf{y})\,(\widehat{\mathbf{x}}_{i}(\mathbf{y})-\mathbf{y}_{i})\geq0,
\end{array}\label{eq:VI_i_row2}
\end{equation}
for all $i\in\mathcal{I}$. \textcolor{black}{Recalling the definition
of $c_{\boldsymbol{{\tau}}}$ {[}cf. (\ref{eq:c_tau}){]}}\textcolor{blue}{{}
}and using (\ref{eq:VI_i_row2}), we obtain \vspace{-0.1cm} 
\begin{equation}
\left(\mathbf{y}_{i}-\widehat{\mathbf{x}}_{i}(\mathbf{y})\right)^{T}\nabla_{\mathbf{x}_{i}}U(\mathbf{y})\geq c_{\boldsymbol{{\tau}}}\left\Vert \widehat{\mathbf{x}}_{i}(\mathbf{y})-\mathbf{y}_{i}\right\Vert ^{2},\label{eq:VI_i_row4}
\end{equation}
for all $i\in\mathcal{I}$. Summing (\ref{eq:VI_i_row4}) over $i$
we obtain (\ref{eq:descent_direction}).

\noindent(a): Let us use the notation as in Lemma \ref{Lemma_f_x_y_properties}.
Given $\mathbf{y},\mathbf{z}\in\mathcal{K}$, by the minimum principle,
we have
\begin{eqnarray}
\left(\mathbf{v}-\widehat{\mathbf{x}}(\mathbf{y})\right)^{T}\nabla_{\mathbf{x}}\tilde{{f}}\left(\widehat{\mathbf{x}}(\mathbf{y});\mathbf{y}\right) & \geq & 0\qquad\forall\mathbf{v}\in\mathcal{K}\nonumber \\
\left(\mathbf{w}-\widehat{\mathbf{x}}(\mathbf{z})\right)^{T}\nabla_{\mathbf{x}}\tilde{{f}}\left(\widehat{\mathbf{x}}(\mathbf{z});\mathbf{z}\right) & \geq & 0\qquad\forall\mathbf{w}\in\mathcal{K}.\label{eq:mp2}
\end{eqnarray}
 Setting $\mathbf{v}=\widehat{\mathbf{x}}(\mathbf{z})$ and $\mathbf{w}=\widehat{\mathbf{x}}(\mathbf{y})$,
summing the two inequalities above, and adding and subtracting $\nabla_{\mathbf{x}}\tilde{{f}}\left(\widehat{\mathbf{x}}(\mathbf{y});\mathbf{z}\right)$,
we obtain:
\begin{equation}
\begin{array}{l}
\left(\widehat{\mathbf{x}}(\mathbf{z})-\widehat{\mathbf{x}}(\mathbf{y})\right)^{T}\left(\nabla_{\mathbf{x}}\tilde{{f}}\left(\widehat{\mathbf{x}}(\mathbf{z});\mathbf{z}\right)-\nabla_{\mathbf{x}}\tilde{{f}}\left(\widehat{\mathbf{x}}(\mathbf{y});\mathbf{z}\right)\right)\\
\leq\left(\widehat{\mathbf{x}}(\mathbf{y})-\widehat{\mathbf{x}}(\mathbf{z})\right)^{T}\left(\nabla_{\mathbf{x}}\tilde{{f}}\left(\widehat{\mathbf{x}}(\mathbf{y});\mathbf{z}\right)-\nabla_{\mathbf{x}}\tilde{{f}}\left(\widehat{\mathbf{x}}(\mathbf{y});\mathbf{y}\right)\right).
\end{array}\label{eq:minimum_principle_Lip}
\end{equation}
 Using (\ref{eq:strong_cvx_f_tilde}) we can now lower bound the left-hand-side
of (\ref{eq:minimum_principle_Lip}) as 
\begin{equation}
\begin{array}{l}
\left(\widehat{\mathbf{x}}(\mathbf{z})-\widehat{\mathbf{x}}(\mathbf{y})\right)^{T}\left(\nabla_{\mathbf{x}}\tilde{{f}}\left(\widehat{\mathbf{x}}(\mathbf{z});\mathbf{z}\right)-\nabla_{\mathbf{x}}\tilde{{f}}\left(\widehat{\mathbf{x}}(\mathbf{y});\mathbf{z}\right)\right)\\
\,\,\geq c_{{\boldsymbol{\tau}}}\left\Vert \widehat{\mathbf{x}}(\mathbf{z})-\widehat{\mathbf{x}}(\mathbf{y})\right\Vert ^{2},
\end{array}\label{eq:lipschtz_map_2}
\end{equation}
whereas the right-hand side of (\ref{eq:minimum_principle_Lip}) can
be upper bounded as 
\begin{equation}
\begin{array}{l}
\left(\widehat{\mathbf{x}}(\mathbf{y})-\widehat{\mathbf{x}}(\mathbf{z})\right)^{T}\left(\nabla_{\mathbf{x}}\tilde{{f}}\left(\widehat{\mathbf{x}}(\mathbf{y});\mathbf{z}\right)-\nabla_{\mathbf{x}}\tilde{{f}}\left(\widehat{\mathbf{x}}(\mathbf{y});\mathbf{y}\right)\right)\\
\,\,\leq\, L_{\nabla\tilde{f}}\,\left\Vert \widehat{\mathbf{x}}(\mathbf{y})-\widehat{\mathbf{x}}(\mathbf{z})\right\Vert \,\left\Vert \mathbf{y}-\mathbf{z}\right\Vert ,
\end{array}\label{eq:lipschtz_map_1}
\end{equation}
where the inequality follows from the Cauchy-Schwartz inequality and
(\ref{eq:Lip_grad_L_f}). Combining (\ref{eq:minimum_principle_Lip}),
(\ref{eq:lipschtz_map_2}), and (\ref{eq:lipschtz_map_1}), we obtain
the desired Lipschitz property of $\widehat{\mathbf{x}}(\bullet)$. 

\noindent (b): Let $\mathbf{x}^{\star}\in\mathcal{K}$ be a fixed
point of $\widehat{\mathbf{x}}(\mathbf{y})$, that is $\mathbf{x}^{\star}=\widehat{\mathbf{x}}(\mathbf{x}^{\star})$.
By definition, each $\widehat{\mathbf{x}}_{i}(\mathbf{y})$ satisfies
(\ref{eq:VI_i}), for any given $\mathbf{y}\in\mathcal{K}$. Setting
$\mathbf{y}=\mathbf{x}^{\star}$ and using $\mathbf{x}^{\star}=\widehat{\mathbf{x}}(\mathbf{x}^{\star})$,
(\ref{eq:VI_i}) reduces to 
\begin{equation}
\left(\mathbf{z}_{i}-\mathbf{x}_{i}^{\star}\right)^{T}\nabla_{\mathbf{x}_{i}}U(\mathbf{x}^{\star})\geq0,\label{eq:fixed_point_min_principle}
\end{equation}
for all $\mathbf{z}_{i}\in\mathcal{K}_{i}$ and $i\in\mathcal{I}$.
Taking into account the Cartesian structure of $\mathcal{K}$ and
summing (\ref{eq:fixed_point_min_principle}) over $i\in\mathcal{I}$
we obtain $\begin{array}[t]{l}
\left(\mathbf{z}-\mathbf{x}^{\star}\right)^{T}\nabla_{\mathbf{x}}U(\mathbf{x}^{\star})\geq0,\end{array}$ for all $\mathbf{z}\in\mathcal{K},$ with $\mathbf{z}\triangleq(\mathbf{z}_{i})_{i=1}^{I}$;
therefore $\mathbf{x}^{\star}$ is a stationary solution of (\ref{eq:social problem}).

The converse holds because i) $\widehat{\mathbf{x}}(\mathbf{x}^{\star})$
is the unique optimal solution of (\ref{eq:decoupled_problem_i})
with $\mathbf{y}=\mathbf{x}^{\star}$, and ii) $\mathbf{x}^{\star}$
is also an optimal solution of (\ref{eq:decoupled_problem_i}), since
it satisfies the minimum principle.

\noindent (d): It follows readily from (\ref{eq:VI_i_row4}).\hfill $\square$\vspace{-0.2cm}

\subsection{Proof of Theorems \ref{Theorem_convergence_Jacobi} and \ref{Theorem_convergence_inexact_Jacobi}\label{sec:Proof-of-Theorem_Jacobi}}

We prove Theorem \ref{Theorem_convergence_inexact_Jacobi}; Theorem
\ref{Theorem_convergence_Jacobi}(b) is a special case; the proof
of simpler Theorem \ref{Theorem_convergence_Jacobi}(a) is omitted
and can be obtained following similar steps.  The line of the proof
is based on standard descent arguments, but suitably combined with
the properties of $\widehat{\mathbf{x}}(\mathbf{y})$ (cf. Prop. \ref{Prop_x_y}),
and the presence of errors $\{\epsilon_{i}^{n}\}$. We will also use
the following lemma, which is the deterministic version of the Robbins-Siegmund
result for random sequences \cite[Lemma 11]{Polyak_book} (but without
requiring the nonnegativity of $X^{n}$ and $Z^{n}$ as instead in
\cite[Lemma 11]{Polyak_book}). \begin{lemma}\label{lemma_Robbinson_Siegmunt}
Let $\{X^{n}\}$, $\{Y^{n}\}$, and $\{Z^{n}\}$ be three sequences
of numbers such that $Y^{n}\geq0$ for all $n$. Suppose that 
\[
X^{n+1}\leq X^{n}-Y^{n}+Z^{n},\quad\forall n=0,1,\ldots
\]
and $\sum_{n}Z^{n}<\infty$. Then either $X^{n}\rightarrow-\infty$
or else $\{X^{n}\}$ converges to a finite value and $\sum_{n}Y^{n}<\infty$.
\hfill $\Box$\end{lemma}

We are now ready to prove Theorem \ref{Theorem_convergence_inexact_Jacobi}.
For any given $n\geq0$, the Descent Lemma \cite{Bertsekas_NLPbook99}
yields 
\begin{equation}
\begin{array}{lll}
U\left(\mathbf{x}^{n+1}\right) & \leq & U\left(\mathbf{x}^{n}\right)+\gamma^{n}\,\nabla_{\mathbf{x}}U\left(\mathbf{x}^{n}\right)^{T}\left(\mathbf{z}^{n}-\mathbf{x}^{n}\right)\smallskip\\
 &  & +\dfrac{\left(\gamma^{n}\right)^{2}{L_{\nabla U}}}{2}\,\left\Vert \mathbf{z}^{n}-\mathbf{x}^{n}\right\Vert ^{2},
\end{array}\label{eq:descent_Lemma}
\end{equation}
with $\mathbf{z}^{n}\triangleq(\mathbf{z}_{i}^{n})_{i=1}^{I}$, and
$\mathbf{z}_{i}^{n}$ defined in Step 2 (Algorithm \ref{alg:PJA-inex}).
Using $\left\Vert \mathbf{z}^{n}-\mathbf{x}^{n}\right\Vert ^{2}\leq2\left\Vert \widehat{\mathbf{x}}(\mathbf{x}^{n})-\mathbf{x}^{n}\right\Vert ^{2}+2\sum_{i}\left\Vert \mathbf{z}_{i}^{n}-\widehat{\mathbf{x}}_{i}(\mathbf{x}^{n})\right\Vert ^{2}\leq2\left\Vert \widehat{\mathbf{x}}(\mathbf{x}^{n})-\mathbf{x}^{n}\right\Vert ^{2}+2\sum_{i}(\varepsilon_{i}^{n})^{2}$,
where in the last inequality we used $\left\Vert \mathbf{z}_{i}^{n}-\widehat{\mathbf{x}}_{i}(\mathbf{x}^{n})\right\Vert \leq\varepsilon_{i}^{n}$,
and 
\begin{equation}
\begin{array}{l}
\nabla_{\mathbf{x}}U\left(\mathbf{x}^{n}\right)^{T}\left(\mathbf{z}^{n}-\widehat{\mathbf{x}}(\mathbf{x}^{n})+\widehat{\mathbf{x}}(\mathbf{x}^{n})-\mathbf{x}^{n}\right)\leq\smallskip\\
\qquad\quad-c_{\boldsymbol{{\tau}}}\left\Vert \widehat{\mathbf{x}}(\mathbf{x}^{n})-\mathbf{x}^{n}\right\Vert ^{2}+\sum_{i}\varepsilon_{i}^{n}\left\Vert \nabla_{\mathbf{x}_{i}}U(\mathbf{x}^{n})\right\Vert ,
\end{array}\label{eq:descent_at_x_n}
\end{equation}
which follows from Prop. \ref{Prop_x_y}(c), (\ref{eq:descent_Lemma})
yields: for all $n\geq0$, 
\begin{equation}
\begin{array}{l}
\!\!\!\!\!\!\!\!\begin{array}{l}
U\left(\mathbf{x}^{n+1}\right)\leq U\left(\mathbf{x}^{n}\right)-\gamma^{n}\left(c_{\boldsymbol{{\tau}}}-\gamma^{n}{L_{\nabla U}}\right)\left\Vert \widehat{\mathbf{x}}(\mathbf{x}^{n})-\mathbf{x}^{n}\right\Vert ^{2}+T_{n},\end{array}\end{array}\label{eq:descent_Lemma_2}
\end{equation}
where $T_{n}\triangleq\gamma^{n}\,\sum_{i}\varepsilon_{i}^{n}\left\Vert \nabla_{\mathbf{x}_{i}}U(\mathbf{x}^{n})\right\Vert +\left(\gamma^{n}\right)^{2}{L_{\nabla U}}\,\sum_{i}(\varepsilon_{i}^{n})^{2}$.
Note that, under the assumptions of the theorem, $\sum_{n=0}^{\infty}T_{n}<\infty$.
Since $\gamma^{n}\rightarrow0$, we have for some positive constant
$\beta_{1}$ and sufficiently large $n$, say $n\geq\bar{{n}}$,
\begin{equation}
U\left(\mathbf{x}^{n+1}\right)\leq U\left(\mathbf{x}^{n}\right)-\gamma^{n}\beta_{1}\left\Vert \widehat{\mathbf{x}}(\mathbf{x}^{n})-\mathbf{x}^{n}\right\Vert ^{2}+T_{n}.\label{eq:descent_Lemma_3_}
\end{equation}
Invoking Lemma \ref{lemma_Robbinson_Siegmunt} with the identifications
$X^{n}=U\left(\mathbf{x}^{n+1}\right)$, $Y^{n}=\gamma^{n}\beta_{1}\left\Vert \widehat{\mathbf{x}}(\mathbf{x}^{n})-\mathbf{x}^{n}\right\Vert ^{2}$
and $Z^{n}=T_{n}$ while using $\sum_{n}T_{n}<\infty$, we deduce
from (\ref{eq:descent_Lemma_3_}) that either $\{U\left(\mathbf{x}^{n}\right)\}\rightarrow-\infty$
or else $\{U\left(\mathbf{x}^{n}\right)\}$ converges to a finite
value and\vspace{-0.1cm} 
\begin{equation}
\lim_{n\rightarrow\infty}\sum_{t=\bar{{n}}}^{n}\gamma^{t}\left\Vert \widehat{\mathbf{x}}(\mathbf{x}^{t})-\mathbf{x}^{t}\right\Vert ^{2}<+\infty.\vspace{-0.1cm}\label{eq:finite_sum_series}
\end{equation}
Since $U(\mathbf{x})$ is coercive, $U(\mathbf{x})\geq\min_{\mathbf{y}\in\mathcal{K}}U(\mathbf{y})>-\infty$,
implying that $\{U\left(\mathbf{x}^{n}\right)\}_{n}$ is convergent;
it follows from (\ref{eq:finite_sum_series}) and $\sum_{n=0}^{\infty}\gamma^{n}=\infty$
that $\liminf_{n\rightarrow\infty}\left\Vert \widehat{\mathbf{x}}(\mathbf{x}^{n})-\mathbf{x}^{n}\right\Vert =0.$

Using Prop. \ref{Prop_x_y}, we show next that $\lim_{n\rightarrow\infty}\left\Vert \widehat{\mathbf{x}}(\mathbf{x}^{n})-\mathbf{x}^{n}\right\Vert =0$;
for notational simplicity we will write $\triangle\widehat{\mathbf{x}}(\mathbf{x}^{n})\triangleq\widehat{\mathbf{x}}(\mathbf{x}^{n})-\mathbf{x}^{n}$.
Suppose, by contradiction, that $\limsup_{n\rightarrow\infty}\left\Vert \triangle\widehat{\mathbf{x}}(\mathbf{x}^{n})\right\Vert >0$.
Then, there exists a $\delta>0$ such that $\left\Vert \triangle\widehat{\mathbf{x}}(\mathbf{x}^{n})\right\Vert >2\delta$
for infinitely many $n$ and also $\left\Vert \triangle\widehat{\mathbf{x}}(\mathbf{x}^{n})\right\Vert <\delta$
for infinitely many $n$. Therefore, one can always find an infinite
set of indexes, say $\mathcal{N}$, having the following properties:
for any $n\in\mathcal{N}$, there exists an integer $i_{n}>n$ such
that 
\begin{eqnarray}
\left\Vert \triangle\widehat{\mathbf{x}}(\mathbf{x}^{n})\right\Vert <\delta, &  & \left\Vert \triangle\widehat{\mathbf{x}}(\mathbf{x}^{i_{n}})\right\Vert >2\delta\medskip\label{eq:outside_interval}\\
\delta\leq\left\Vert \triangle\widehat{\mathbf{x}}(\mathbf{x}^{j})\right\Vert \leq2\delta &  & n<j<i_{n}.\label{eq:inside_interval}
\end{eqnarray}
Given the above bounds, the following holds: for all $n\in\mathcal{N}$,
\begin{eqnarray}
\delta & \overset{(a)}{<} & \left\Vert \triangle\widehat{\mathbf{x}}(\mathbf{x}^{i_{n}})\right\Vert -\left\Vert \triangle\widehat{\mathbf{x}}(\mathbf{x}^{n})\right\Vert \medskip\nonumber \\
 & \leq & \left\Vert \widehat{\mathbf{x}}(\mathbf{x}^{i_{n}})-\widehat{\mathbf{x}}(\mathbf{x}^{n})\right\Vert +\left\Vert \mathbf{x}^{i_{n}}-\mathbf{x}^{n}\right\Vert \\
 & \overset{(b)}{\leq} & (1+\hat{{L}})\left\Vert \mathbf{x}^{i_{n}}-\mathbf{x}^{n}\right\Vert \\
 & \overset{(c)}{\leq} & (1+\hat{{L}})\sum_{t=n}^{i_{n}-1}\gamma^{t}\left(\left\Vert \triangle\widehat{\mathbf{x}}(\mathbf{x}^{t})\right\Vert +\left\Vert \mathbf{z}^{t}-\widehat{\mathbf{x}}(\mathbf{x}^{t})\right\Vert \right)\vspace{-0.3cm}\nonumber \\
 & \overset{(d)}{\leq} & (1+\hat{{L}})\,(2\delta+\varepsilon^{\max})\sum_{t=n}^{i_{n}-1}\gamma^{t},\label{eq:lower_bound_sum}
\end{eqnarray}
where (a) follows from (\ref{eq:outside_interval}) and (\ref{eq:inside_interval});
(b) is due to Prop. \ref{Prop_x_y}(a); (c) comes from the triangle
inequality and the updating rule of the algorithm; and in (d) we used
(\ref{eq:outside_interval}), (\ref{eq:inside_interval}), and $\left\Vert \mathbf{z}^{t}-\widehat{\mathbf{x}}(\mathbf{x}^{t})\right\Vert \leq\sum_{i}\varepsilon_{i}^{t}$,
where $\varepsilon^{\max}\triangleq\max_{n}\sum_{i}\varepsilon_{i}^{n}<\infty$.
It follows from (\ref{eq:lower_bound_sum}) that 
\begin{equation}
\liminf_{n\rightarrow\infty}\sum_{t=n}^{i_{n}-1}\gamma^{t}\geq\dfrac{{\delta}}{(1+\hat{{L}})(2\delta+\varepsilon^{\max})}>0.\label{eq:lim_inf_bound}
\end{equation}

We show next that (\ref{eq:lim_inf_bound}) is in contradiction with
the convergence of $\{U(\mathbf{x}^{n})\}_{n}$. To do that, we preliminary
prove that, for sufficiently large $n\in\mathcal{N}$, it must be
$\left\Vert \triangle\widehat{\mathbf{x}}(\mathbf{x}^{n})\right\Vert \geq\delta/2$.
Proceeding as in (\ref{eq:lower_bound_sum}), we have: for any given
$n\in\mathcal{N}$,
\[
\begin{array}{l}
\left\Vert \triangle\widehat{\mathbf{x}}(\mathbf{x}^{n+1})\right\Vert -\left\Vert \triangle\widehat{\mathbf{x}}(\mathbf{x}^{n})\right\Vert \leq(1+\hat{{L}})\left\Vert \mathbf{x}^{n+1}-\mathbf{x}^{n}\right\Vert \smallskip\\
\qquad\qquad\qquad\qquad\qquad\leq(1+\hat{{L}})\gamma^{n}\left(\left\Vert \triangle\widehat{\mathbf{x}}(\mathbf{x}^{n})\right\Vert +\varepsilon^{\max}\right).
\end{array}
\]
 It turns out that for sufficiently large $n\in\mathcal{N}$ so that
$(1+\hat{{L}})\gamma^{n}<\delta/(\delta+2\varepsilon^{\max})$, it
must be 
\begin{equation}
\left\Vert \triangle\widehat{\mathbf{x}}(\mathbf{x}^{n})\right\Vert \geq\delta/2;\label{eq:lower_bound_delta_x_n}
\end{equation}
otherwise the condition $\left\Vert \triangle\widehat{\mathbf{x}}(\mathbf{x}^{n+1})\right\Vert \geq\delta$
would be violated {[}cf. (\ref{eq:inside_interval}){]}. Hereafter
we assume w.l.o.g. that (\ref{eq:lower_bound_delta_x_n}) holds for
all $n\in\mathcal{N}$ (in fact, one can alway restrict $\{\mathbf{x}^{n}\}_{n\in\mathcal{N}}$
to a proper subsequence).

We can show now that (\ref{eq:lim_inf_bound}) is in contradiction
with the convergence of $\{U(\mathbf{x}^{n})\}_{n}$. Using (\ref{eq:descent_Lemma_3_})
(possibly over a subsequence), we have: for sufficiently large $n\in\mathcal{N}$,
\begin{eqnarray}
U\left(\mathbf{x}^{i_{n}}\right) & \leq & U\left(\mathbf{x}^{n}\right)-\beta_{2}\sum_{t=n}^{i_{n}-1}\gamma^{t}\left\Vert \triangle\widehat{\mathbf{x}}(\mathbf{x}^{t})\right\Vert ^{2}+\sum_{t=n}^{i_{n}-1}T_{t}\nonumber \\
 & \overset{(a)}{<} & U\left(\mathbf{x}^{n}\right)-\beta_{2}(\delta^{2}/4)\sum_{t=n}^{i_{n}-1}\gamma^{t}+\sum_{t=n}^{i_{n}-1}T_{t}\label{eq:liminf_zero}
\end{eqnarray}
where in (a) we used (\ref{eq:inside_interval}) and (\ref{eq:lower_bound_delta_x_n}),
and $\beta_{2}$ is some positive constant. Since $\{U(\mathbf{x}^{n})\}_{n}$
converges and $\sum_{n=0}^{\infty}T_{n}<\infty$, (\ref{eq:liminf_zero})
implies $\lim_{\mathcal{N}\ni n\rightarrow\infty}\,\sum_{t=n}^{i_{n}-1}\gamma^{t}=0,$
which contradicts (\ref{eq:lim_inf_bound}).

Finally, since the sequence $\{\mathbf{x}^{n}\}$ is bounded {[}due
to the coercivity of $U(\mathbf{x})$ and the convergence of $\{U(\mathbf{x}^{n})\}_{n}${]},
it has at least one limit point $\bar{{\mathbf{x}}}$ that must belong
to $\mathcal{K}$. By the continuity of $\widehat{\mathbf{x}}(\bullet)$
{[}Prop. \ref{Prop_x_y}(a){]} and $\lim_{n\rightarrow\infty}\left\Vert \widehat{\mathbf{x}}(\mathbf{x}^{n})-\mathbf{x}^{n}\right\Vert =0$,
it must be $\widehat{\mathbf{x}}(\bar{{\mathbf{x}}})=\bar{{\mathbf{x}}}$.
By Prop. \ref{Prop_x_y}(b) $\bar{{\mathbf{x}}}$ is also a stationary
solution of the social problem (\ref{eq:social problem}). 

Note that, in the setting of Theorem \ref{Theorem_convergence_Jacobi},
$\varepsilon_{i}^{n}=0$ for all $i$ and $n$; therefore $T_{n}=0$
for all $n$. It follows from (\ref{eq:descent_Lemma_3_}) that $U(\mathbf{x}^{n})$
is a decreasing sequence, which entails that no limit point of $\{\mathbf{x}^{n}\}$
can be a local maximum. \hfill$\square$\vspace{-0.1cm}

\subsection{Proof of Theorem \ref{Theorem_convergence_GS} \label{sec:Proof-of-Theorem_GS}}

The main idea of the proof is to interpret Algorithm \ref{alg:PGSA}
as an instance of the inexact Jacobi scheme described in Algorithm
\ref{alg:PJA-inex}, and show that Theorem \ref{Theorem_convergence_inexact_Jacobi}
is satisfied. It is not difficult to show that this reduces to prove
that, for all $i=1,\ldots,I$, the sequence $\mathbf{z}_{i}^{t}$
in Step 2a) of Algorithm \ref{alg:PGSA} satisfies 
\begin{equation}
\|\mathbf{z}_{i}^{t}-\widehat{\mathbf{x}}_{i}(\mathbf{x}^{t})\|\leq\tilde{\varepsilon}_{i}^{\, t},\label{eq:GS_Jacobi_error_bound}
\end{equation}
for some $\{\tilde{\varepsilon}_{i}^{\, t}\}$ such that $\sum_{t}\tilde{\varepsilon}_{i}^{\, t}\,\gamma^{t}<\infty$.
The following holds for the LHS of (\ref{eq:GS_Jacobi_error_bound}):
\[
\begin{array}{l}
\!\!\!\!\|\mathbf{z}_{i}^{t}-\widehat{\mathbf{x}}_{i}(\mathbf{x}^{t})\|\!\leq\!\|\widehat{\mathbf{x}}_{i}(\mathbf{x}_{i<}^{t+1},\mathbf{x}_{i\geq}^{t})-\widehat{\mathbf{x}}_{i}(\mathbf{x}^{t})\|\!+\!\|\mathbf{z}_{i}^{t}-\widehat{\mathbf{x}}_{i}(\mathbf{x}_{i<}^{t+1},\mathbf{x}_{i\geq}^{t})\|\smallskip\end{array}
\]
\vspace{-0.5cm} 
\[
\begin{array}{l}
\quad\quad\overset{(a)}{\leq}\|\widehat{\mathbf{x}}_{i}(\mathbf{x}_{i<}^{t+1},\mathbf{x}_{i\geq}^{t})-\widehat{\mathbf{x}}_{i}(\mathbf{x}^{t})\|+{\varepsilon}_{i}^{\, t}\medskip\\
\quad\quad\overset{(b)}{\le}\hat{{L}}\,\|\mathbf{x}_{i<}^{t+1}-\mathbf{x}_{<i}^{t}\|+{\varepsilon}_{i}^{\, t}\medskip\\
\quad\quad\overset{(c)}{\leq}\hat{{L}}\gamma^{t}\left(\left\Vert \left(\widehat{\mathbf{x}}_{j}(\mathbf{x}_{j<}^{t+1},\mathbf{x}_{j\geq}^{t})-\mathbf{x}_{j}^{t}\right)_{j=1}^{i-1}\right\Vert +\sum_{j<i}{\varepsilon}_{j}^{\, t}\right)+{\varepsilon}_{i}^{\, t}\\
\quad\quad\overset{(d)}{\leq}\hat{{L}}\gamma^{t}\beta_{i}+\hat{{L}}\gamma^{t}\sum_{j<i}{\varepsilon}_{j}^{\, t}+{\varepsilon}_{i}^{\, t},
\end{array}
\]
where (a) follows from the error bound in Step 2a) of Algorithm \ref{alg:PGSA};
in (b) we used Prop. \ref{Prop_x_y}a); (c) follows from Step 2b);
and in (d) we used Prop. \ref{Prop_x_y}d), with $\beta_{i}<\infty$
being a positive constant. It turns out that (\ref{eq:GS_Jacobi_error_bound})
is satisfied choosing $\tilde{\varepsilon}_{i}^{\, t}\triangleq\hat{{L}}\gamma^{t}\beta_{i}+\hat{{L}}\gamma^{t}\sum_{j<i}{\varepsilon}_{j}^{\, t}+{\varepsilon}_{i}^{\, t}$.
\hfill$\square$

\bibliographystyle{IEEEtran}
\bibliography{scutari_refs}

\end{document}